\documentclass[reqno]{amsart}
\usepackage{amsmath,amssymb,cite}
\usepackage[mathscr]{euscript}
\theoremstyle{plain}
\newtheorem{theorem}{Theorem}[section]

\newtheorem{lemma}[theorem]{Lemma}

\theoremstyle{definition}

\theoremstyle{remark}

\numberwithin{equation}{section}
\numberwithin{theorem}{section}

\def\be{\begin{equation}}
\def\ee{\end{equation}}
\def\bp{\begin{pmatrix}}
\def\ep{\end{pmatrix}}
\def\bea{\begin{eqnarray}}
\def\eea{\end{eqnarray}}

\def\\{\par\medskip}

\let\0=\noindent

\newcommand{\mc}[1]{{\mathcal #1}}

\newcommand{\bb}[1]{{\mathbb #1}}

\newcommand{\rme}{\mathrm{e}}

\newcommand{\rmd}{\mathrm{d}}

 \let\b=\beta    \let\g=\gamma
\let\eps=\varepsilon        
        \let\x=\xi         
\let\s=\sigma     \let\ph=\varphi

\newcommand{\id}{{1 \mskip -5mu {\rm I}}}
\renewcommand{\epsilon}{\varepsilon}
\renewcommand{\tilde}{\widetilde}
\renewcommand{\hat}{\widehat}

\newcommand{\asinh}{\mathop{\rm arcsinh}\nolimits}
\newcommand{\atanh}{\mathop{\rm arctanh}\nolimits}

\title[On large deviations of interface motions]
{On large deviations of interface motions for statistical mechanics models}

\author[L.\ Bertini]{Lorenzo Bertini}
\address{Lorenzo Bertini \hfill\break \indent
   Dipartimento di Matematica, 
   Sapienza Universit\`a di Roma 
   \hfill\break \indent
   P.le Aldo Moro 5, 00185 Roma, Italy}
 \email{bertini@mat.uniroma1.it}

\author[P.\ Butt\`a]{Paolo Butt\`a}
\address{Paolo Butt\`a\hfill\break \indent
   Dipartimento di Matematica, 
   Sapienza Universit\`a di Roma  
   \hfill\break \indent
   P.le Aldo Moro 5, 00185 Roma, Italy}
 \email{butta@mat.uniroma1.it}

\author[A.\ Pisante]{Adriano Pisante}
\address{Adriano Pisante\hfill\break \indent
Dipartimento di Matematica, 
Sapienza Universit\`a di Roma  
\hfill\break \indent
P.le Aldo Moro 5, 00185 Roma, Italy}
 \email{pisante@mat.uniroma1.it}

\begin{document}

\begin{abstract}
We discuss the sharp interface limit of the action functional associated to either the Glauber dynamics for Ising systems with Kac potentials or the Glauber+Kawasaki process. The corresponding limiting functionals, for which we provide explicit formulae of the mobility and transport coefficients, describe the large deviations asymptotics with respect to the mean curvature flow.
\end{abstract}
%\keywords{}
\maketitle
\thispagestyle{empty}

\section{Introduction}
\label{sec:1}

Consider the dynamical evolution, with non conserved order parameter, of a system undergoing a first order phase transition. A basic paradigm of statistical mechanics is that the corresponding macroscopic behavior is described by the motion by curvature of the interfaces separating the two stable phases. For lattice systems with short range interaction, the lattice symmetries are still felt on the macroscopic scale and the resulting evolution is an anisotropic motion by curvature. For values of the temperature below the roughening transition, the Wulff shape is not strictly convex and the corresponding evolution is crystalline, i.e., it generates facets \cite{TCH,T}. On the other hand, for long range interactions, the resulting interface evolution is described by the (isotropic) motion by mean curvature. We refer to \cite{funaki} for a recent overview on stochastic interface evolutions. 

In principle, the macroscopic evolution of the interfaces should be derived from a microscopic Glauber-like dynamics, and the corresponding transport coefficients could be characterized in terms of the microscopic interaction and the jump rates. While there is plenty of numerical evidence that this is indeed the case, the analytical results are few and the derivation of motion by curvature, say for the Ising model with Glauber dynamics at positive temperature, remains a most challenging issue. For short range interactions, the only available results are in fact at zero temperature \cite{CMST,LST,S}. In the case of long range interactions, or more precisely for Ising model with Kac potentials, the motion by mean curvature has been derived in \cite{DOPT1,KS}. The peculiar feature of this model is the presence of a parameter, the interaction range, that allows to achieve this derivation in two separate steps. Firstly, it is considered the evolution of the empirical magnetization in the Lebowitz-Penrose limit, showing that its limiting behavior is described by a non local evolution equation. Secondly, it is shown that, under a diffusive rescaling of space and time, such evolution leads to the motion by mean curvature. This second step is quite similar to the analogous derivation starting from the Allen-Cahn equation \cite{BSS,I}. There is another model with the same features, the so-called Glauber+Kawasaki process, for which the derivation of motion by mean curvature has been achieved by the same procedure \cite{Bo,KS0}.

The present purpose is to describe, in the sense of large deviations theory, the probability of deviations from the motion by curvature. Postponing the connection with the microscopic dynamics, let us first discuss this topic purely from a phenomenological point of view in the setting introduced in \cite{S}. On a scale large compared to the microscopic length scale, we can represent the interface between the two pure phases as a surface $\Gamma$ of codimension one embedded in $\bb R^d$. The typical evolution of $\Gamma$ can then be deduced by free energy considerations. We denote by $\tau$ the surface tension, that in general depends on the local orientation of the surface, i.e., on the local normal $\hat n$ at $\Gamma$. The surface free energy is then given by
\begin{equation}
	\label{free}
	F = \int_{\Gamma} \!\rmd\sigma\, \tau (\hat n)\;,
\end{equation}
where $\rmd \sigma$ is the surface measure. Observe that in the isotropic case $\tau$ is constant and $F$ becomes proportional to the perimeter of $\Gamma$. Phenomenologically, it is postulated that the interface velocity along the local normal, denoted by $v$, is given by
\begin{equation}
\label{mbc0}
v = -\mu \frac {\delta F}{\delta \Gamma}\;,
\end{equation}
where the \emph{mobility} $\mu$ may depend on the local orientation on the surface. As shown in \cite{S}, for short range interactions the mobility $\mu$ can be computed from the microscopic dynamics, by either a Green-Kubo formula obtained via a linear response argument, or by looking at the fluctuations of the empirical order parameter.

Let $\tilde\tau$ be the 1-homogeneous extension of $\tau$ to a function on $\bb R^d$, and introduce the stiffness matrix $A(\hat n)$ as the Hessian of $\tilde\tau$ at $\hat n$ (so that $A(\hat n) \hat n=0$). For $x\in\Gamma$ we define,
\[
\kappa_A (x) := \tau(\hat n(x))^{-1} \sum_{i=1}^{d-1} \langle
e_i(x),A(\hat n(x))e_i(x)\rangle \kappa_i(x)\;,
\]
where $\kappa_i(x)$ are the principal curvatures and $e_i(x)$ are the corresponding principal curvature directions of $\Gamma$ at $x$. Then \eqref{mbc0} reads,
\[
v = \theta \kappa_A\;,
\]
where the \emph{transport coefficient} $\theta$ is given by the Einstein relation,
\begin{equation}
\label{theta}
\theta = \mu \tau\;.
\end{equation}
In the isotropic case, $\tau$ and $\mu$ are constant, $A(\hat n) = \tau \id$ on the subspace orthogonal to $\hat n$,  hence $\kappa_A = \kappa$, the mean curvature of $\Gamma$.

Referring to \cite{S} for the analysis of (small) Gaussian fluctuations, we next introduce the rate function describing the asymptotics of the probability of large deviations around the motion by mean curvature. To this end, fix a time interval $[0,T]$ and a path $\Gamma(t)$, $t\in[0,T]$. On a basis of a Gaussian assumption on the noise and a fluctuation dissipation relation, the rate function ought to be given by
\begin{equation}
\label{sac}
 S_\mathrm{ac}(\Gamma) = \frac 1{4\mu} \int_0^T\!\rmd t\int_{\Gamma(t)}\rmd \sigma\, (v-\theta \kappa_A)^2\;. 
\end{equation}
This functional should catch the asymptotics of the probability of smooth paths, and we next discuss its extension to more general paths. As shown in \cite{KORV} in the context of the Allen-Cahn equation, the path $t\mapsto \Gamma(t)$ need not to be continuous since nucleation might occur at some intermediate times. In such cases, the appropriate rate function reads,
\begin{equation}
\label{s+s}
S(\Gamma) = S_\mathrm{ac}(\Gamma) + S_\mathrm{nucl}(\Gamma) \;,
\end{equation}
where $S_\mathrm{nucl}$ measures, according to \eqref{free}, the free energy cost of the interfaces nucleated in the time interval $[0,T]$. As we discuss in Section \ref{sec:5}, $S_\mathrm{nucl}$ can be recovered from $S_\mathrm{ac}$ by approximating nucleation events with continuous  paths. Moreover, interfaces need to be counted with their multiplicity and are not necessarily smooth, even away from the nucleation times. Suitable weak definitions of the curvature and velocity are thus needed. This  is accomplished by using tools of geometric measure theory, we refer to \cite{MR} for the proper definition of the functional $S$  in the case of non smooth interfaces in the isotropic case. Finally, it cannot be excluded that the map $t\mapsto \Gamma(t)$ has a Cantor part, which does not affect the cost functional constructed in \cite{MR}. A variational definition of $S$ which takes into account also such Cantor part is provided in \cite{BBP1}, its corresponding zero level set is given by the mean curvature flow according to the Brakke's formulation \cite{Brakke}. 

The rate functional $S$ should describe the large deviations asymptotics of microscopic stochastic dynamics that leads to the motion by curvature of the interfaces. The corresponding analysis has been carried out mostly for the Allen-Cahn evolution. In particular, the functional \eqref{s+s} has been identified by considering the sharp interface limit of the natural action functional associated to the Allen-Cahn equation, initially in \cite{KORV} and in greater detail in \cite{MR}.  A stochastic Allen-Cahn equation has been considered in \cite{BBP1}, where it is proven the large deviation upper bound with rate function $S$.  Observe that, as discussed in \cite{S}, the Allen-Cahn evolution exhibits a trivial transport coefficient, $\mu=1/\tau$, so that $\theta=1$ regardless of the shape of the double well potential. The case of Glauber dynamics for Ising systems with Kac potentials, in the one dimensional case, has been considered in \cite{BT,BDT}, where it is evaluated the asymptotic probability of a displacement of an interface in a given finite time.

Here we discuss, in the case of smooth interfaces, the derivation of the rate function $S$ by considering either the Glauber dynamics for Ising systems with Kac potentials or the Glauber+Kawasaki process. For these models the large deviations asymptotics, respectively in the Lebowitz-Penrose and in the continuum limit, has been derived in \cite{C} and in \cite{JLV}. We thus analyze the sharp interface limit of the corresponding action functionals, deducing the rate functional \eqref{sac} and providing explicit formulae for the mobility coefficients. While the basic strategy is analogous to the one in \cite{KORV}, the non local character of the action functionals requires a more clever choice of the optimizing sequences. More precisely, in order to obtain the right transport coefficient, we need to introduce a corrector in the ansatz for the recovery sequences and solve a variational problem to identify the optimal choice. In the case of the Ising model with Kac potentials, the mobility derived here agrees with that derived in \cite{B} by a linear response argument, thus validating the fluctuations dissipation assumption. The computation of the mobility for the Glauber+Kawasaki process appears instead novel and provides a dynamical characterization of the surface tension. Note indeed that, as the invariant measure of this process is not explicitly known, a static characterization according to the guidelines of equilibrium statical mechanics is not feasible. 

It would be interesting to extend the results of the present paper on the rate function $S$ to the case of general interfaces, possibly exhibiting nucleation events. In analogy with the results in \cite{MR} for the Allen-Cahn equation, a key step should be to describe in both the models considered here the asymptotic behavior of sequences $\varphi_\eps$ with equibounded action (see the equations \eqref{S1} and \eqref{Ihk2} below). In the case of Ising-Kac a compactness property is expected, in analogy with the result in \cite{AB} for time-independent sequences with equibounded free energy (see \eqref{Fe} below), yielding paths of sharp interfaces in the limit $\eps \to 0$ with uniformly bounded perimeter. However, it is unclear how to associate to such configurations $\varphi_\eps$  corresponding paths of generalized surfaces $t \mapsto \Gamma_\eps(t)$ (as varifolds in the sense of geometric measure theory) with suitable uniform curvature and velocity bounds. As a consequence, we are not able to deduce curvature and velocity bounds on the limiting interfaces. Moreover, it remains to be proven that the well-prepared sequences $\varphi_\eps$ here considered (see \eqref{recovery} and \eqref{recovery1}) actually describe the typical asymptotic behavior of configurations assuming uniquely boundedness of the action functionals.  

\section{ Glauber dynamics with Kac potentials}
\label{sec:2}

In this section we analyze the sharp interface limit of the action functional in the context of the Glauber dynamics for Ising systems with Kac potentials.

\subsection{Microscopic model and its mean field limit.} 
\label{sec:2.0}

Let $\bb T^d_L= (\bb R/ L\bb Z)^d$ be the torus of side $L\ge 1$ in $\bb R^d$; when $L=1$ we drop it from the notation, i.e., $\bb T^d = \bb T^d_1$. We denote by $r,r'$ the elements of $\bb T^d_L$ and by $\rmd r$ the Haar measure on $\bb T^d_L$. Given a smooth non-negative function $j\colon \bb R_+ \to \bb R_+$, supported in $[0,\frac 12]$ and such that $\int_{\bb R^d}\!\rmd z\; j(|z|) = 1$, we let $J\colon \bb T^d_L \to \bb R_+$ be the probability density defined by $J(r) =  j(|r|)$. In the sequel, $J*f(r) := \int_{\bb T^d_L} \!\rmd r'\, J(r-r') f(r')$ is the standard convolution on $\bb T^d_L$.

Given $L>0$, and $\gamma>0$ such that $\gamma^{-1}L \in \bb N$, let $\bb T^d_{L,\gamma} := (\gamma \bb Z / L \bb Z)^d $ be the discrete approximation of $\bb T^d_L$ with lattice spacing $\gamma$. The microscopic configuration space is $\Omega_{L,\gamma} := \{-1,1\}^{\bb T^d_{L,\gamma}}$. The microscopic energy is the function $H_\gamma \colon \Omega_{L,\gamma} \to \bb R$ defined by 
\[
H_{L,\gamma}(\s) = -\frac 12 \sum_{i,j\in \bb T^d_{L,\gamma}} \gamma^d J(i-j) \sigma(i)\sigma(j)\;.  
\]
Given the inverse temperature $\beta > 0$, the corresponding Gibbs measure $\mu_{L,\gamma}^\beta$ is the probability on $\Omega_{L,\gamma}$ defined by 
\begin{equation}
\label{gibbs}
\mu_{L,\gamma}^\beta(\sigma) = \frac 1{Z^\beta_{L,\gamma}} \exp \big\{-\beta H_{L,\gamma}(\sigma)\big\}\;,
\end{equation}
 where $Z^\beta_{L,\gamma}$ is the partition function.
 
 \subsubsection*{Lebowitz-Penrose limit}
We consider the supercritical case $\beta>1$ and define the spontaneous magnetization $m_\beta$ as the strictly positive solution of the Curie-Weiss equation, that is
\begin{equation}
\label{mbeta}
m_\beta = \tanh (\beta m_\beta)\;, \quad m_\beta>0\;.
\end{equation}
Denoting by $\mc M(\bb T^d_L)$ the space of bounded measures on the torus $\bb T^d_L$, equipped with the weak*-topology,
we define the \emph{empirical magnetization} as the map $M^\gamma\colon \Omega_{L,\gamma} \to  \mc M(\bb T^d_L)$ given by 
\[
M^\gamma(\sigma) = \gamma^d \sum_{i\in \bb T^d_{L,\gamma}} \sigma(i) \, \delta_i\; .
\]
 
As proven in \cite{EE}, in the Lebowitz-Penrose limit $\gamma\to 0$ the excess free energy functional for the Gibbs measures \eqref{gibbs} is the functional $F_L\colon L^\infty(\bb T^d_L;[-1,1]) \to [0,\infty)$ given by
\be
\label{F}
F_L(m) = \int\!\rmd r\; [f_\beta(m)-f_\beta(m_\b)] + \frac 14 \int\!\rmd r\!\int\!\rmd r'\; J(r-r') [m(r)-m(r')]^2\;,
\ee
where
\be
\label{fb}
f_\beta(m) = - \frac{m^2}2 + \b^{-1} \imath(m)\;, \quad \imath(m) = \frac{1+m}2\log\frac{1+m}2 + \frac{1-m}2\log\frac{1-m}2\;.
\ee
Observe that, since $\pm m_\beta$ are the minimizers of $f_\beta$, the functional $F_L$ vanishes on the pure phases $\pm m_\beta$. The probabilistic content of this statement is that the family $\{\mu^\gamma_{L,\gamma} \circ (M^\gamma)^{-1} \}_{\gamma>0}$ of probabilities on $ \mc M(\bb T^d_L)$ satisfies a large deviation principle with speed $\beta^{-1}\gamma^d$ and rate function $\mc F_L$ given by $\mc F_L(\nu)=F_L(m)$ if $\nu= m \, \rmd r$ for some $m \in L^\infty(\bb T^d_L;[-1,1])$ and $+\infty$ otherwise.
 
\subsubsection*{Glauber-Kac dynamics} 
The Glauber dynamics with Kac potentials is a continuous time Markov chain on the state space $\Omega_{L,\gamma}$, reversible with respect to the Gibbs measure \eqref{gibbs}. It is defined by assigning the rates at which the value of the spin $\sigma$ at site $i$ is flipped. The corresponding generator $\bb L_\gamma$ is the operator acting on functions on $\Omega_{L,\gamma}$ as
\be
\label{glauber}
\bb L_\g f(\s) = \sum_{\substack{i\in \bb T^d_{L,\gamma}}} c(i,M^\gamma(\sigma)) \rme^{-\beta  J*M^\gamma(\sigma) (i) \s(i)} [f(\s^i)-f(\s)]\;,
\ee
where $c \colon \bb T^d_L\times \mc M(\bb T^d_L) \to (0,+\infty)$ is a continuous function satisfying  $c(r, \nu) = c(r, \nu-\nu(\{r\})\delta_r)$, which implies the detailed balance condition, namely, that $\bb L_\gamma$ is self-adjoint with respect to the Gibbs measure \eqref{gibbs}.

In order to perform the sharp interface limit, we restrict to a special class of rates. More precisely, we assume that 
\begin{equation}
\label{ctype}
c(r,\nu) \equiv c(\nu)(r) = a(K*\nu(r))\;, \quad r\in \bb T^d_L\;,
\end{equation}
where $a\colon \bb R \to (0,+\infty)$ is a Lipschitz function and $K$ is a smooth radial function on $\bb R^d$ with support in the ball of radius $\frac 12$ and satisfying $K(0)=0$; i.e., $K(r) = k(|r|)$ for a smooth non-negative function $k\colon \bb R_+ \to \bb R_+$ with support in $[0,\frac 12]$. A standard choice, see \cite{DOPT1}, is 
\be
\label{c}
c(i,M^\gamma(\sigma)) = \frac{1}{2\cosh\big\{\b \sum_{j\ne i} \gamma^d J(i-j) \sigma(j)\big\}}\;,
\ee
that, provided $J(0)=0$, corresponds to $c(r,\nu) = (2\cosh \b J* \nu(r))^{-1}$.

\subsubsection*{Mean field evolution equation} 
As proven in \cite{DOPT1}, in the Lebowitz-Penrose limit (mesoscopic limit) the empirical magnetization under the Glauber dynamics becomes absolutely continuous and its density $m$  evolves according to the non-local equation,
\be
\label{mfe1}
\frac{\partial m}{\partial t} = - 2c(m) \sqrt{1-m^2} \sinh (\atanh m -\b J*m)\;,
\ee
We notice that expanding the $\sinh$, Eq.\eqref{mfe1}  reads,
\be
\label{mfe2}
\frac{\partial m}{\partial t} = 2c(m) \cosh(\b J*m) (\tanh(\b J*m) -m)\;.
\ee
In particular, with the choice \eqref{c}, the mean field evolution becomes,
\be
\label{mfee}
\frac{\partial m}{\partial t} = \tanh(\b J*m) -m\;.
\ee

The stationary solutions to \eqref{mfe1} do not depend on the particular choice of the rates. In particular, since we are assuming $\b>1$, recalling \eqref{mbeta}, the spatially homogeneous stationary solutions are $m=\pm m_\b$, that are stable, and $m=0$, which is unstable.

\subsubsection*{Action functional}
The large deviation asymptotics for the empirical magnetization under the Glauber dynamics for an Ising spin system with Kac potentials has been analyzed in \cite{C}. We next recall the associated rate function. 

Let $B_1(L)$ be the unit ball in $L^\infty(\bb T_L^d)$ equipped with the (metrizable) weak*-topology. For $T>0$ we then let $C([0,T];B_1(L))$  be the set of $B_1(L)$-valued continuous functions equipped with the induced uniform distance. Let finally $C_*([0,T];B_1(L))$ be the subset of functions $\varphi$ in $C([0,T];B_1(L))$ such that there exists $\psi \in L^1([0,T] \times \bb T_L^d)$ for which
\[
\varphi(t,r)-\varphi(0,r)=\int_0^t \psi(s,r) \, \rmd s \quad \text{$r$ - a.e.} \quad \forall\, t \in [0,T]\;.
\]
Clearly $\psi$ is unique and will be denoted by $\dot{\varphi}$. We define $I_{T,L} \colon C([0,T];B_1(L)) \to [0, \infty]$ by
 
\be
\label{I}
I_{T,L}(\varphi)= 
\begin{cases}
\displaystyle{\int_0^T\!\rmd t\int \!\rmd r\; \mc L(\ph(t, \cdot),\dot\ph(t,\cdot))} & \text{if } \varphi \in C_*([0,T];B_1(L)) \, , \\ \\
+\infty & \text{otherwise} \, ,
\end{cases}
\ee
where, given measurable functions $u \colon \bb T_L^d \to [-1,1]$ and $v \colon \bb T_L^d \to \bb R$,   
\begin{equation}
\label{mc L=}
\begin{split}
 & \mc L(u,v) = \frac v{2\beta} \log\frac{\displaystyle \tfrac{v}{2c(u)} + \sqrt{1-u^2+\tfrac{v^2}{4c(u)^2}}}{1-u} - \frac v2 J*u \\ & \quad + \frac{c(u)}\beta \Big(\cosh(\b J*u) - u \sinh(\b J*u) - \sqrt{1-u^2+\tfrac{v^2}{4c(u)^2}}\,\Big) \;.
\end{split}
\end{equation}
Under suitable assumptions on the initial conditions, in \cite{C} it is proven that the empirical magnetization sampled according to the  Glauber dynamics, regarded as a random variable taking values in the Skorokhod space $D([0,T];\mc M(\bb T^d_L))$, satisfies a large deviation principle with speed $\beta^{-1} \gamma^d$ and rate function $\mc I_{T,L}$ given by $\mc I_{T,L}(\nu) = I_{T,L}(\ph)$ if $\nu_t= \varphi_t \, \rmd r$ for some $\varphi\in C_*([0,T];B_1(L))$ and $+\infty$ otherwise. 

For our purposes, by noticing that, as $\iota'(m) = \atanh m$, the functional derivative ($L^2$-gradient) of $F_L$ is given by 
\be
\label{dF}
\frac{\delta F_L}{\delta m} = \b^{-1} \atanh m -J*m\;,
\ee
we rewrite the Lagrangian $\mc L$ in \eqref{mc L=} in the form, 
\[
\begin{split}
\mc L(u,v) & = \frac v{2\beta} \left( \atanh u  - \b J*u + \asinh \frac{v}{2c(u)\sqrt{1-u^2}} \right) \\ & \quad + \frac{c(u)}\beta \sqrt{1-u^2} \bigg(\cosh\big(\b J*u -\atanh u\big) - \sqrt{1+\tfrac{v^2}{4c(u)^2(1-u^2)}}\bigg) \\ & = \frac v2 \left(\frac{\delta F_L}{\delta u} + \frac1\beta\asinh \frac{v}{2c(u)\sqrt{1-u^2}} \right) \\ & \quad + \frac{c(u)}\beta\sqrt{1-u^2} \left(\cosh\left(\b\frac{\delta F_L}{\delta u} \right) - \sqrt{1+\tfrac{v^2}{4c(u)^2(1-u^2)}}\right) \;.
\end{split}
\]
Accordingly, the action functional becomes,
\be
\label{IS}
\begin{split}
& I_{T,L}(\varphi)= \frac 12 \big[ F_L(\ph(T)) - F_L(\ph(0))\big]  \\ & \;+ \int_0^T\!\rmd t\int\!\rmd r\; \bigg[ \frac{\dot\ph}{2\beta} \asinh W(\ph,\dot\varphi)- \frac{c(\ph)}\beta\sqrt{1-\ph^2} \big(\sqrt{1+W(\ph,\dot\ph)^2}-1\big) \bigg] \\  & \;+ \int_0^T\!\rmd t\int\!\rmd r\; \frac{c(\ph)}\beta\sqrt{1-\ph^2} \left(\cosh\left(\b\frac{\delta F_L}{\delta \varphi} \right) - 1\right)\;.
\end{split}
\ee
where
\be
\label{W}
W(\ph,\dot\varphi)= \frac{\dot\ph}{2c(\ph)\sqrt{1-\ph^2}}\;.
\ee

It is worthwhile to remark that the above representation of the action functional reflects a Legendre duality. More precisely, for $\alpha>0$ let $G(\cdot;\alpha)$ and $G^*(\cdot;\alpha)$ be the Legendre pair of convex even functions,
\begin{equation}
\label{gg*=}
G(q;\alpha) := \alpha(\cosh q -1)\;, \quad G^*(p;\alpha) = p \asinh(p/\alpha) - \sqrt{\alpha^2 +p^2} + \alpha\;,
\end{equation}
so that $qp + G(q;\alpha) + G^*(p;\alpha) \ge 0$ with equality if and only if $p=-\alpha\sinh q$. Then \eqref{IS} can be rewritten as
\begin{equation}
\label{gg*}	
\begin{split}
I_{T,L}(\varphi)& = \frac 12 \big[ F_L(\ph(T)) - F_L(\ph(0))\big] \\ & \quad + \frac 12\int_0^T\!\rmd t\int\!\rmd r\; \Big[G\big(\b\tfrac{\delta F_L}{\delta \varphi};\alpha(\ph)\big) + G^*\big(\beta^{-1}\dot\ph;\alpha(\ph)\big)\Big]\;,
\end{split}
\end{equation}
where $\alpha(\varphi)= 2 \beta^{-1} c(\varphi)\sqrt{1-\varphi^2}$. From this representation we easily conclude that the solution $m$ to the mean field equation \eqref{mfe1} is characterized by $I_{T,L}(m) = 0$, or equivalently $I_{T,L}(m) \le 0$. The last inequality provides the following gradient flow formulation: $m$ is a solution to \eqref{mfe1} if and only if, for any $t\in [0,T]$, 
\[
F_L(m(t)) + \int_0^t\!\rmd s\int\!\rmd r\; \Big[G\big(\b\tfrac{\delta F_L}{\delta m};\alpha(m)\big) + G^*\big(\beta^{-1}\dot m;\alpha(m)\big)\Big] \le F_L(m(0))\;.
\] 

\subsection{Sharp interface limit} 
\label{sec:3}

A natural and physically relevant question is to investigate the limiting behavior of the Ising-Kac model in the sharp interface limit, in which the interface between the two stable phases $\pm m_\beta$ is described by surfaces of codimension one.

\subsubsection*{Excess free energy and surface tension}
We set $\eps=L^{-1}$, and rescale the space variable $r \in \bb T^d_L$ by setting $r=\eps^{-1}x$ with $x\in \bb T^d$. We then introduce the  rescaled excess free energy renormalized with a factor $L^{d-1}$. We namely define $F^{\eps} \colon L^\infty(\bb T^d;[-1,1]) \to [0,\infty)$ by $F^\eps(m) := \eps^{d-1} F_{\eps^{-1}}(m(\eps^{-1} \cdot))$, i.e.,
\be
\label{Fe}
F^\eps(m) = \int \!\rmd x\; \frac{f_\beta(m)-f_\beta(m_\b)}{\eps} + \frac \eps 4 \int \!\rmd x\!\int \!\rmd y\; J_\eps(x-y) \left[ \frac{m(x)-m(y)}{\eps}\right]^2\;,
\ee
where $J_\eps(z) := \eps^{-d}J(\eps^{-1}z)$. The asymptotics of the excess free energy functional \eqref{Fe} has been discussed in \cite{AB,ABCP}, where it is proven that the limiting functional is finite only if $m$ takes the values $\pm m_\beta$, and in this case its value is proportional to the perimeter of the jump set of $m$. The proportionality factor defines the \emph{surface tension} of the Ising-Kac model, which is denoted by $\tau$ and will be characterized below. This result has been extended to the anisotropic case, i.e., when $J$ is not radial; then, the surface tension $\tau$ is no longer constant but a convex function of the orientation \cite{AB,BBP}.

The surface tension is the excess free energy cost per unit area of the transition between the two stable phases. The characterization of $\tau$ reduces to a one dimensional computation in the direction normal to the interface. We introduce the \emph{instanton} $\bar m(\xi)$, $\xi\in \bb R$, as the optimal magnetization profile of such a transition, that is, $\bar m$ is solution to
\begin{equation}
\label{inst}
\bar m(\xi) = \tanh \b \tilde J * \bar m (\xi)\;, \quad \bar m(0) = 0\;, \quad \lim_{\xi\to\pm\infty}\bar m(\xi) = \pm m_\b\;, 
\end{equation}
where, recalling  $J(r) =  j(|r|)$,
\begin{equation}
\label{tildeJ}	
\tilde J(\xi) = \int_{\bb R^{d-1}}\!\rmd \eta\; j\big(\sqrt{\xi^2+|\eta|^2}\big) \;.
\end{equation}
Then $\tau = \mc F (\bar m)$, where $\mc F$ is the free energy functional on $\bb R$, 
\be
\label{mcF}
 \mc F(m) = \int\!\rmd \xi\; [f_\beta(m)-f_\beta(m_\b)] + \frac 14 \int\!\rmd \xi\!\int\!\rmd \xi'\; \tilde J(\xi-\xi') [m(\xi)-m(\xi')]^2\;.
 \ee
 It can be shown \cite{B} that
 \be
 \label{tau}
 \tau = \int\!\rmd\xi \; \bar m'(\xi) \int\!\rmd\xi' \; \int_{\bb R^{d-1}}\!\rmd \eta\; j\big(\sqrt{(\xi-\xi')^2+|\eta|^2}\big) \; \bar m'(\xi') \; \frac{\eta_1^2}2\;.
 \ee
 
 For later purpose, we recall the main properties of the instanton, see \cite{DOPT2,DOPT3,DGP}. It is an odd and strictly increasing function which converges exponentially fast to its asymptotes. More precisely,  $\bar m'(\xi)>0$ and there are $a,c,\delta>0$ such that, for any $\xi\ge 0$,
 \begin{equation}
 	\label{mexp}
 	\big| \bar m(\xi) - (m_\beta-a\rme^{-\alpha\xi}) \big| + \big| \bar m'(\xi) - a\alpha \rme^{-\alpha\xi}) \big| + \big| \bar m''(\xi) - a\alpha^2 \rme^{-\alpha\xi}) \big| \le c \rme^{-(\alpha+\delta)\xi}\;,
 \end{equation}
 where $\alpha$ is the unique positive solution to the equation
 \begin{equation}
 	\label{alpha}
 	\beta(1-m_\beta^2)\int\!\rmd\xi\; \tilde J(\xi)	\rme^{-\alpha\xi}  = 1\;.
 \end{equation}
 
 \subsubsection*{Motion by mean curvature}
 Concerning the dynamical behavior, the sharp interface limit of the nonlocal evolution equation has been analyzed in \cite{DOPT,DOPT1,KS}, with the special choice of $c$ as in \eqref{c}. To describe these results, let $m$ be the solution to \eqref{mfee} and define, according to a diffusive rescaling of space and time, $m^\eps \colon \bb R_+ \times \bb T^d \to [-1,1]$ by $m^\eps(t,x) = m(\eps^{-2}t,\eps^{-1}x)$, which solves
\be
\label{mfeee}
\frac{\partial m^\eps}{\partial t} = \eps^{-2}\big(\tanh(\b J_\eps*m^\eps) -m^\eps\big)\;.
\ee 
In order to describe the limiting behavior of $m^\eps$, we briefly recall the notion of classical mean curvature flow. Given a $C^1$-family of oriented smooth surfaces $\Gamma=\{ \Gamma(t)\}_{t\geq 0}$, with $\Gamma(t) = \partial\Omega(t)$ for some open  $\Omega(t) \subset \bb T^d$,
we denote by $n_t=n_{\Gamma(t)}$ the inward normal of $\Gamma(t)$, by $v_t\colon \Gamma(t) \to \bb R$ the normal velocity of $\Gamma$ at time $t$. Finally, we set $\kappa_t=\kappa_{\Gamma(t)}$, where $\kappa_{\Gamma(t)} \colon \Gamma(t) \to \bb R$ is the mean curvature of $\Gamma(t)$. Then, given $\theta>0$, $\Gamma$ evolves according to the mean curvature flow with transport coefficient $\theta>0$ if
\be
\label{mbc}
v_t=\theta \kappa_t \, , \qquad t\geq 0 \, . 
\ee 
Given a mean curvature flow as above, assuming that the initial datum for \eqref{mfee} satisfies $m^\eps(0,\cdot) \to m_\beta\id_{\Omega(0)}-m_\beta \id_{\Omega(0)^\complement}$, then $m^\eps(t,\cdot) \to m_\beta\id_{\Omega(t)}-m_\beta \id_{\Omega(t)^\complement}$ for any $t>0$. The actual value of $\theta$ obtained in \cite{DOPT,KS} will be discussed later. 

In \cite{DOPT1,KS} the convergence to the mean curvature flow is proven also starting directly from the microscopic Glauber dynamics. More precisely, letting $M^{\gamma,\eps}$ be the diffusively rescaled empirical magnetization, it is shown that if $\eps = |\log\gamma|^{-1}$ then $M^{\gamma,\eps}$ satisfies the law of large numbers as $\gamma\to 0$, and the limiting evolution is given by the mean curvature flow.
 
\subsubsection*{Transport coefficients and Einstein relation} 
The value of the transport coefficient $\theta$, for arbitrary $c(m)$ of the form \eqref{ctype}, can be inferred by using a linear response argument along the guidelines in \cite{S}. Consider the non local mean field equation \eqref{mfe2} on $\bb R^d$ with external field $h$, that is,
\[
\frac{\partial m}{\partial t} = 2c(m) \cosh(\b(J*m+h)) [\tanh(\b (J*m+h))-m]\;.
\] 
In view of \eqref{ctype} and recalling that $J$ and $K$ are radial, solutions to the above equation with planar symmetry along a fixed direction $\hat n$ have the form $m(t,\eta) = \tilde m(\eta\cdot \hat n,t)$ with $\tilde m(\xi,t)$, $\xi\in\bb R$, solution to
\begin{equation}
\label{mfe1d0}
\frac{\partial \tilde m}{\partial t} =  2 a(\tilde K *\tilde m) \cosh(\b(\tilde J *\tilde m+h)) [\tanh(\b (\tilde J *\tilde m+h))-\tilde m]\;,
\end{equation}
where $\tilde J$ is defined in \eqref{tildeJ} and, analogously, recalling $K(r)=k(|r|)$,
\begin{equation}
\label{tildeK}	
\tilde K(\xi) = \int_{\bb R^{d-1}}\!\rmd \eta\; k\big(\sqrt{\xi^2+|\eta|^2}\big) \;.
\end{equation}
In particular, if we look for a traveling wave solution along $\hat n$, i.e., a solution of the form $m(t,\eta) = q_h(\eta\cdot \hat n-v(h)t)$, we deduce that $q_h$ and the front velocity $v(h)$ do not depend on the direction $\hat n$ and solve (in the case of \eqref{mfee} with $h$ small their existence is proven in \cite{DGP})
\begin{equation}
\label{mfe1d}
-v(h) q_h' =  2 a(\tilde K *q_h) \cosh(\b(\tilde J *q_h+h)) [\tanh(\b (\tilde J *q_h+h))-q_h]\;.
\end{equation}

In order to compute the linear response to the external field we expand,
\[
v(h) = v_1h + O(h^{2})\;,\qquad q_h = \bar m + h \psi  + O(h^{2}) \;,
\]
where $\bar m$ is the instanton which solves \eqref{mfe1d} with $h=0$ and $v(0)=0$, see \eqref{inst}. In the sequel we set 
\begin{equation}
	\label{cxi}
	\bar a(\xi) := a(\tilde K*\bar m(\xi))\:, \quad \xi\in \bb R\;.	
\end{equation}
By \eqref{mfe1d}, at the first order in $h$, we obtain the following identity,
 \[
 - v_{1} \bar m' = \frac{2\bar a}{\sqrt{1-\bar m^2}} \big[ - \psi + (1- \bar m^2) \b \tilde J * \psi + \b (1 - \bar m^2) \big]\;,
 \]
 where we used that $\cosh(\b\tilde J*\bar m) = 1/\sqrt{1-\tanh^2(\b\tilde J*\bar m)} = 1/\sqrt{1-\bar m^2}$. We multiply both sides of the above equation by $\bar m'/(2\bar a \sqrt{1-\bar m^2})$ and then integrate; using that $\bar m' = (1-\bar m^2)\b \tilde J *\bar m'$ we  obtain,
 \[ 
 v_{1} = - 2N \beta m_{\beta}\;,
 \]
 where 
 \be
 \label{N}
 N = \left[\int\!\rmd\xi\; \frac{(\bar m')^2}{2\bar a\sqrt{1-\bar m^2}}\right]^{-1}.
 \ee
 But, by the definition of the (macroscopic) mobility $\mu$, see \cite{S}, it must be $v(h) = - 2m_{\beta} \mu h + O(h^{2}) $. We conclude that 
 \be
 \label{mu}
 \mu = N \beta\;.
 \ee
 We finally remark that in the case \eqref{c} we have $2\bar a = \sqrt{1-\bar m^2}$, so that $N = \left[\int\!\rmd\xi\; \frac{(\bar m')^2}{1-\bar m^2}\right]^{-1}$ in this case.
  
\subsubsection*{Sharp interface limit of the action functional} 
The main purpose of the section is to discuss the sharp interface limit of the action functional. To this end, we perform a diffusive rescaling of space and time of parameter $\eps=L^{-1}$ and normalize the resulting action with a factor $L^{d-1}$. Namely, given $T>0$, we define $S_\eps \colon C([0,T];B_1) \to [0,\infty]$ (here $B_1$ is a short notation  for the unit ball $B_1(1)$ in $L^\infty(\bb T^d)$) by
\be
\label{S}
S_\eps(\varphi)= \eps^{d-1} I_{\eps^{-2}T,\eps^{-1}}(\ph(\eps^2 \cdot, \eps \cdot ))= \eps^{-1} \int_0^T \rmd t \int \rmd x \, \mc L_\eps (\varphi(t,\cdot), \dot{\varphi}(t,\cdot))\;,
\ee
where, given measurable functions $u \colon \bb T^d \to [-1,1]$, $v \colon \bb T^d \to \bb R$ and recalling $J_\eps(\cdot):= \eps^{-d} J(\cdot /\eps)$,
\begin{equation}
  \label{mc Leps=}
  \begin{split}
  & \mc L_\eps(u,v) = \frac {v}{2\beta} \log\frac{\displaystyle \frac{\eps^{2} v}{2c_\eps(u)} +
    \sqrt{1-u^2+\left(\frac{\eps^{2} v}{2c_\eps(u)}\right)^2}}{1-u} - \frac { v}2
  J_\eps*u \\ & \quad + \frac{c_\eps(u)}{\beta\eps^2} \left(\cosh(\b J_\eps*u) - u \sinh(\b J_\eps*u) -
    \sqrt{1-u^2+\left(\frac{\eps^{2} v}{2c_\eps(u)}\right)^2}\right)\;,
\end{split}
\end{equation}
with, recalling \eqref{ctype} and letting $K_\eps(\cdot):= \eps^{-d} K(\cdot /\eps)$,
\begin{equation}
\label{ceps}
c_\eps(u) := a(K_\eps*u)\;.
\end{equation}

Given a $C^1$-family of oriented smooth surfaces $\Gamma=\{ \Gamma(t)\}_{t\in [0,T]}$, with $\Gamma(t) = \partial\Omega(t)$ for some open  $\Omega(t) \subset \bb T^d$, as before we denote by $n_t=n_{\Gamma(t)}$ the inward normal of $\Gamma(t)$, by $v_t\colon \Gamma(t) \to \bb R$ the normal velocity of $\Gamma$ at time $t$, and by $\kappa_t$ the mean curvature of $\Gamma(t)$. Letting $\tilde d(\cdot,\Gamma(t))$ be the signed distance from $\Gamma(t)$, i.e., $\tilde d(\cdot,\Gamma(t)) := \mathrm{dist}(\cdot,\Omega(t)^\complement) -  \mathrm{dist}(\cdot, \Omega(t))$, we denote by $d(\cdot,\Gamma(t))$ a regularized version of $\tilde d(\cdot,\Gamma(t))$ such that they coincide on a neighborhood of $\Gamma(t)$. 

For such families of surfaces we consider the action functional \eqref{sac}, i.e.,
\be
\label{Rpn}
S_\mathrm{ac}(\Gamma) = \frac 1{4\mu}\int_0^T\!\rmd t\int_{\Gamma(t)}\!\rmd\s\; (v_t-\theta \kappa_t)^2\;,
\ee
with $\mu$ as given in \eqref{mu} and $\theta=\mu\tau$ with $\tau$ as defined in \eqref{tau}. As next stated, it describes the sharp interface limit of the rescaled action functional associated to the Glauber dynamics for an Ising system with Kac potentials. 

\begin{theorem}
\label{thm:3.1}
Given a $C^1$-family of oriented smooth surfaces $\Gamma=\{ \Gamma(t)\}_{t\in [0,T]}$, with $\Gamma(t) = \partial\Omega(t)$, consider sequences $\{ \varphi_\eps \} \subset C([0,T];B_1)$ converging to $m_\beta\id_{\Omega(\cdot)}-m_\beta \id_{\Omega(\cdot)^\complement}$ of the form  
\begin{equation}
\label{recovery}
\varphi_\eps(t,x) = \bar m\left(\frac{d(x,\Gamma(t))}\eps + \eps Q\left(t,x,\frac{d(x,\Gamma(t))}\eps\right)\right)+\eps R_\eps(t,x)\;,
\end{equation}
where $\bar{m}$ is the instanton, $Q \colon [0,T] \times \bb T^d \times \bb R \to \bb R$ is a smooth function such that
\be
\label{Q}
\sup_{(t,x,\xi) \in [0,T] \times \bb T^d \times \bb R}  \frac{\big| Q(t,x,\xi) \big| +  \big|\partial_t Q(t,x,\xi) \big| +\big| \partial_\xi Q(t,x,\xi)\big|}{1+|\xi|} <+\infty\;,
\ee
and $R_\eps \colon [0,T] \times \bb T^d \to \bb R$ is a smooth function.
\begin{itemize}
\item[(a)] If $\| R_\eps\|_\infty +\| \partial_t R_\eps \|_\infty \to 0$ as $\eps \to 0$  then, for any $Q$,
\[
\liminf_{\eps\to 0} S_\eps(\varphi_\eps) \ge S_\mathrm{ac}(\Gamma)\;.
\]
\item[(b)] There exist $Q^*$ such that, choosing $Q=Q^*$ and $R_\eps=0$ we have,
\[
\lim_{\eps\to 0} S_\eps(\varphi_\eps) = S_\mathrm{ac}(\Gamma)\;.
\]
\end{itemize}
\end{theorem}

From a physical viewpoint, the main content of the result is the identification of the transport coefficients in the limiting rate function $S$. As expected, the mobility $\mu$, that is initially introduced via a linear response argument, coincides with the variance of the fluctuations around the motion by mean curvature. The mechanism behind this identification is an averaging property, common to homogenization problems. At the mathematical level, this is achieved by the introduction (in the same spirit of \cite{KS}) of the corrector $Q$ in the ansatz \eqref{recovery}: the transport coefficients are then identified by solving an optimization problem on $Q$. As mentioned in the Introduction, this issue does not appear in the Allen-Cahn case, in which the introduction of correctors is not needed.

The above statement is the analogous for the Ising-Kac model of that in \cite[Prop.\ 2.2]{KORV} for the Allen-Cahn action functional. While these results hint to the variational convergence (more precisely $\Gamma$-convergence) of the sequence of functionals $S_\eps$ to $S_\mathrm{ac}$, from technical viewpoint there are several missing steps. Concerning the lower bound, i.e., statement (a) in the theorem, the main difficulty consists in showing that the sequences $\varphi_\eps$ satisfying $S_\eps(\varphi_\eps)\le C$ are of the one-dimensional form given by \eqref{recovery} for suitable (not necessarily smooth) path $\Gamma$ and some $Q$ and $R_\eps$. In the Allen-Cahn case this is proven in \cite{MR}, where this structure of the sequence $\varphi_\eps$ is deduced as a consequence of the vanishing property of the discrepancy measures. As we have no analogue of the discrepancy measures for the Ising-Kac model, we have no clue on how to handle this issue in the present case. Concerning the upper bound, statement (b) provides the construction of the recovery sequence when the limiting path $\Gamma$ is smooth without nucleations. Combining, via a diagonal argument, this statement with the argument presented in Section \ref{sec:5}, it is also possible to construct a recovery sequence for piecewise $C^1$-paths. The missing step, that is common with the Allen-Cahn case, is the proof that general paths of finite action can be approximated by piecewise $C^1$-paths.

A natural further step is the analysis of the large deviation properties of the empirical magnetization for the underlying microscopic dynamics in the joint limit $\gamma\to 0$ and $\eps\to 0$, for instance when $\eps = |\log\gamma|^{-1}$. For the stochastic Allen-Cahn equation, the large deviations upper bound, with rate function $S$, is proven in \cite{BBP1} by constructing suitable exponential martingales. This strategy seems applicable also to the Ising-Kac model, but requires, as a crucial step, the $\Gamma$-convergence lower bound discussed above.

\subsection{Proof of Theorem \ref{thm:3.1}} 
\label{sec:3b}

To carry out the proof, we shall need the following results on the linearization of the nonlocal evolution. Consider Eq.\eqref{mfe1d0} for $h=0$; by \eqref{cxi} and using again the identity $\cosh(\b\tilde J*\bar m) = 1/\sqrt{1-\bar m^2}$,  the linearization around the instanton gives rise to the linear operator,
\be
\label{L}
L\psi = \frac{2\bar a}{\sqrt{1-\bar m^2}} (-\psi + (1-\bar m^2)\b \tilde J*\psi)\;. 
\ee
We regard it as an operator on $L^2(\bb R,\nu(\rmd\xi))$, where
\be
\label{nu}
\nu(\rmd\xi) = \frac{\rmd\xi}{2 \bar a(\xi) \sqrt{1-\bar m^2(\xi)}}\;.
\ee
We observe that $L$ is bounded, symmetric, and negative semidefinite, with $0$ a simple eigenvalue and $\bar m'$ the corresponding eigenvector. In fact, using again that $\bar m' = (1-\bar m^2)\b \tilde J *\bar m'$, it is easy to check that $L\bar m' = 0$ and that
\[
\int\!\nu(\rmd\xi)\; \psi(\xi) L\psi(\xi) = -\frac 12 \int\!\rmd \xi\int\!\rmd \xi'\; \b\tilde J(\xi-\xi') \bar m'(\xi)\bar m'(\xi') \left[\frac{\psi}{\bar m'}(\xi) - \frac{\psi}{\bar m'}(\xi')\right]^2 \;.
\]
As $\tilde J(0)>0$ and $\tilde J$ is continuous, we infer that the integral on the right-hand side is zero if and only if $\psi/\bar m'$ is constant.
An application of Weyl's theorem shows that $L$ has the gap property, i.e., that $0$ is an isolated eigenvalue.  The above arguments can be found in \cite{DOPT2} for the case \eqref{c}. A similar result holds also in $L^\infty$. This is done in \cite{DGP} for the case \eqref{c}, the extension to the general case is straightforward. 

For expository reasons, we prove the statements in reverse order.

\noindent {\it Proof of \rm (b).} 
Recalling \eqref{IS}, the decomposition \eqref{gg*}, and \eqref{gg*=}, we rewrite the rescaled action functional \eqref{S} as
\be
\label{S1}
S_\eps(\varphi) = S_\eps^{(1)}(\varphi)+  S_\eps^{(2)}(\varphi) +  S_\eps^{(3)}(\ph)\;,
\ee
where
\[
\begin{split}
S_\eps^{(1)}(\varphi)& = \frac 12 \big[ F^\epsilon(\ph(T)) - F^\epsilon(\ph(0))\big]\;, \\ S_\eps^{(2)}(\varphi) & = \frac 12 \int_0^T\!\rmd t\int \!\rmd x\; G^*((\beta\eps)^{-1}\dot \ph;\alpha_\eps(\ph)) \;, \\ S_\eps^{(3)}(\varphi) & = \frac 12 \int_0^T\!\rmd t\int \!\rmd x\; G\big(\eps\b\tfrac{\delta F^\epsilon}{\delta m} (\varphi);\alpha_\eps(\ph)\big)\;,
\end{split}
\]
with $F^\epsilon$ as in \eqref{Fe} and
\begin{equation}
	\label{Weps}
	\alpha_\eps(\varphi)= 2 \frac{c_\eps(\ph)}{\beta\eps^3} \sqrt{1-\ph^2} \;, \quad \eps\b\frac{\delta F^\epsilon}{\delta m}(\varphi) = \atanh \varphi -\b J_\eps*\varphi\;.
\end{equation}

In the sequel we choose $\ph=\varphi_\eps$ as in \eqref{recovery}, with $R_\eps =0$ and $Q$ to be determined later, and analyze separately the contribution of the three terms in \eqref{S1}. 

\noindent
1) As proven in \cite{P}, the free energy $F^\epsilon$ $\Gamma$-converges to $\tau \mathrm{Per}(\cdot)$, where $\mathrm{Per}(\cdot)$ is the perimeter functional. Moreover, for any choice of the corrector $Q$ and $t\in [0,T]$,  the function $\varphi_\eps(t,\cdot)$ is a recovery sequence. Hence,
\[
\lim_{\eps\to 0} S_\eps^{(1)}(\varphi_\eps)  = \frac{\tau}2 \left[\mathrm{Per}(\Omega(T)) - \mathrm{Per}(\Omega(0)\right] = -\frac{\tau}2 \int_0^T\!\rmd t \int_{\Gamma(t)} \!\rmd\s\; \kappa_t v_t\;,
\]
where in the last equality we used that  $-\int_{\Gamma(t)} \!\rmd\s\, \kappa_t v_t$ is the time derivative of $\mathrm{Per}(\Omega(t)) $. By \eqref{theta} we thus have,
\be
\label{s1}
\lim_{\eps\to 0} S_\eps^{(1)}(\varphi_\eps)  = -\frac1\mu \int_0^T\!\rmd t \int_{\Gamma(t)} \!\rmd\s\; \frac{\theta\kappa_t v_t}2\;.
\ee

\noindent
2) We first notice that, by Taylor expansion, $G^*(p,\alpha) = \alpha \big[\tfrac 12 \big(\tfrac p\alpha\big)^2 + O\big(\big(\tfrac p\alpha\big)^4\big)\big]$. By \eqref{Q},
\begin{equation}
\label{phid}
\dot\varphi_\eps (x,t) = -\tfrac {\partial_td(x,\Gamma(t))}\eps \; \bar m'\left(\tfrac{d(x,\Gamma(t))}\eps(1+O(\eps)) \right) \left(1+\left(1+\tfrac{|d(x,\Gamma(t))|}\eps\right) O(\eps)\right)\;.
\end{equation}
As $\bar m'(\xi)$ converges exponentially fast to zero as $|\xi|\to\infty$, see \eqref{mexp}, and in view of \eqref{Weps}, the integrand appearing in $S_\eps^{(2)}$ is smaller than any power of $\eps$ if $|d(x,\Gamma(t))|>C\eps(\log\eps)^2$. Therefore, we can restrict the domain of integration in a small neighborhood of $\Gamma_t$. In view of the expansion of $G^*$, using the co-area formula, we then get,
\[
\lim_{\eps \to 0} S_\eps^{(2)}(\varphi_\eps)  = \lim_{\eps \to 0} \int_0^T\!\rmd t\int\limits_{|s|\le C\eps(\log\eps)^2}\!\rmd s\int_{d=s}\!\rmd\s\; \eps^{-1}\;\frac{\bar m'(s/\eps)^2}{2c_\eps(\varphi_\eps)\sqrt{1-\bar m(s/\eps)^2}}\; \frac {(\partial_t d)^2}{4\beta}\;,
\]
where $\rmd\sigma$ is the surface measure on the level set of the distance function $d$, and, by \eqref{ceps},  $c_\eps(\varphi_\eps) = a(K_\eps*\varphi_\eps)$. To compute the asymptotic behavior of $c_\eps(\varphi_\eps)$, we choose an orthonormal frame with origin in the orthogonal projection $x_{\Gamma(t)}$ of $x$ on $\Gamma(t)$ and the first direction $\mathrm{e_0}$ along the normal to $\Gamma(t)$ at $x_{\Gamma(t)}$. If $d(x,\Gamma(t)) = s$ then $x = s\; \mathrm{e_0}$ and therefore, using \eqref{tildeK},
\[
\begin{split}
K_\eps*\varphi_\eps(x)  & = \int\!\rmd y\; \eps^{-d} k\left(\eps^{-1}\sqrt{(s-y\cdot \mathrm{e_0})^2+|y-(y\cdot \mathrm{e_0})\mathrm{e_0}|^2}\right) \bar m \left(\frac{y\cdot \mathrm{e_0}}\eps\right) + O(\eps) \\  & =  \int\!\rmd\xi'\; \tilde K\left(\frac{s}\eps - \xi'\right) \bar m(\xi') + O (\eps) = \tilde K*\bar m \left(\frac{s}\eps \right)  + O(\eps)\;.
\end{split} 
\]
We conclude that, recalling the definition of $\bar a$ in \eqref{cxi},
\be
\label{s2}
\lim_{\eps \to 0} S_\eps^{(2)}(\varphi_\eps)  = \int_0^T\!\rmd t\int_{\Gamma(t)} \!\rmd\s\; \frac{v_t^2}{4\beta}\int\!\rmd \xi\; \frac{\bar m'(\xi)^2}{2\bar a(\xi) \sqrt{1-\bar m(\xi)^2}} = \frac 1\mu \int_0^T\!\rmd t\int_{\Gamma(t)} \!\rmd\s\; \frac{v_t^2}4\;,
\ee
where we used that $-\partial_t d(\cdot,\Gamma(t)) = v_t$ on $\Gamma(t)$, and \eqref{N} and \eqref{mu} in the last identity.

\noindent
3) We are left with the limit of $S_\eps^{(3)}(\varphi_\eps)$. This is the point where the corrector $Q$ plays a role and has to be chosen appropriately. As $\atanh \bar m = \b \tilde J *\bar m$, 
\[
\begin{split}
\eps\b\frac{\delta F^\epsilon}{\delta m}(\varphi_\eps) (x,t)& = \b \int\!\rmd\xi'\; \tilde J\left(\frac{d(x,\Gamma(t))}\eps + \eps Q\left(t,x,\frac{d(x,\Gamma(t))}\eps \right) - \xi'\right) \bar m(\xi') \\ & \qquad - \b \int\!\rmd y\; J_\eps(x,y) \; \bar m \left(\frac{d(y,\Gamma(t))}\eps + \eps Q\left(t,x,\frac{d(y,\Gamma(t))}\eps \right)\right)\;.
\end{split} 
\]
Since $\bar m(\xi)$ converges exponentially fast to $\pm m_\b$ as $\xi\to\pm\infty$, see \eqref{mexp}, the above expression is smaller than any power of $\eps$ if $|d(x,\Gamma(t))|>C\eps(\log\eps)^2$. Therefore, as $G(q,\alpha) = \alpha \big[\tfrac 12 q^2 + O(q^4)\big]$, restricting the domain of integration, and using the previous computation for the limit of $c_\eps(\varphi_\eps)$, we obtain,
\[
\begin{split}
& \lim_{\eps \to 0}S_\eps^{(3)}(\varphi_\eps) \\  &  = \lim_{\eps \to 0}  \frac1{2\beta\eps^3} \int_0^T\!\rmd t\int\limits_{|s|\le C\eps(\log\eps)^2}\!\rmd s \int_{d=s}\!\rmd\s\; \bar a(s/\eps)\sqrt{1-\bar m(s/\eps)^2} \; \left(\eps\b\frac{\delta F^\epsilon}{\delta m}(\varphi_\eps)\right)^2. 
\end{split}
\]
To compute $\eps\b\frac{\delta F^\epsilon}{\delta m}(\varphi_\eps)$ we choose an orthonormal frame with origin in the orthogonal projection $x_{\Gamma(t)}$ of $x$ on $\Gamma(t)$, the first direction $\mathrm{e_0}$ along the normal to $\Gamma(t)$, and the remaining directions $\{\mathrm{e}_1,\ldots,\mathrm{e}_{d-1}\}$ along the principal curvature directions of $\Gamma(t)$. In this way, if $d(x,\Gamma(t)) = s$ with $|s|\le C\eps(\log\eps)^2$ and $|x-y|\le \eps$,  we have,
\[
x = s\; \mathrm{e_0}\;, \qquad d(y,\Gamma(t)) =y\cdot \mathrm{e_0} - \sum_{i=1}^{d-1} \kappa^{(i)}_t \frac{(y\cdot \mathrm{e}_i)^2}2+o(\eps^2)\;, 
\]
where $\kappa^{(i)}_t$ are the principal curvatures of $\Gamma(t)$ at $x_{\Gamma(t)}$; in particular,  the mean curvature reads $\kappa_t = \sum_{i=1}^{d-1} \kappa^{(i)}_t$. Therefore, if $d(x,\Gamma(t)) = s$,
\[
\begin{split}
\b \int\!\rmd\xi'\; & \tilde J\left(\frac{d(x,\Gamma(t))}\eps + \eps Q\left(t,x,\frac{d(x,\Gamma(t))}\eps \right) - \xi'\right) \bar m(\xi') \\ & = \b \int\!\rmd\xi'\; \tilde J\left(\frac{s}\eps - \xi'\right) \left[\bar m(\xi') + \eps Q\left(t,x,\frac{s}\eps \right) \bar m'(\xi')\right] + o(\eps)
\end{split}
\]
and
\[
\begin{split}
& \b \int\!\rmd y\; J_\eps(x,y) \; \bar m \left(\frac{d(y,\Gamma(t))}\eps + \eps Q\left(t,y,\frac{d(y,\Gamma(t))}\eps \right)\right) \\ &\;\; = \b \int\!\rmd y\; \eps^{-d} j\left(\eps^{-1}\sqrt{(s-y\cdot \mathrm{e_0})^2+|y-(y\cdot \mathrm{e_0})\mathrm{e_0}|^2}\right) \bigg[\bar m \left(\frac{y\cdot \mathrm{e_0}}\eps\right)  \\ & \qquad +\eps Q\left(t,x,\frac{y\cdot \mathrm{e_0}} \eps \right)\bar m'\left(\frac{y\cdot \mathrm{e_0}}\eps\right) - \sum_{i=1}^{d-1} \kappa^{(i)}_t \frac{(y\cdot \mathrm{e}_i)^2}2\;\bar m'\left(\frac{y\cdot \mathrm{e_0}}\eps\right)\bigg] + o (\eps) \\ & \;\; = \b \int\!\rmd\xi'\; \tilde J\left(\frac{s}\eps - \xi'\right) \big[\bar m(\xi') + \eps Q(t,x,\xi')\bar m'(\xi')\big] \\ &\qquad - \eps \b \kappa_t \int\!\rmd\xi'\int_{\bb R^{d-1}}\!\rmd\eta\; j\left(\sqrt{\left(\frac s\eps -\xi'\right)^2+|\eta|^2}\right) \bar m'(\xi')\; \frac{\eta_1^2}2 + o (\eps)\;.
\end{split} 
\]

We now choose $Q(t,x,\xi) = Q^*(t,x,\xi) := \mc K(t,x) \bar{Q}(\xi)$, where $\mc K \colon [0,T] \times \bb T^d \to \bb R$ is any smooth function satisfying $\mc K(t, x)= \kappa_t(x)$ for all $x \in \Gamma(t)$,  while $\bar{Q} \colon \bb R \to \bb R$ is a suitable a smooth function, to be fixed later and satisfying
\be
\label{Q1}
\sup_{\xi\in\bb R} \frac{\big| \bar Q(\xi) \big| + \big|\bar Q'(\xi)\big|}{1+|\xi|}<+\infty\;.
\ee
Therefore, under this assumption,
\be
\label{edf}
\begin{split}
	\eps\b\frac{\delta F^\epsilon}{\delta m}(\varphi_\eps) & = \eps\b \mc K(t,x)\int\!\rmd\xi'\; \tilde J\left(\frac{s}\eps - \xi'\right) \bar m'(\xi') \left[ \bar{Q}\left(\frac{s}\eps \right) -\bar{Q}(\xi')\right] \\ & +\eps \b \kappa_t \int\!\rmd\xi'\int_{\bb R^{d-1}}\!\rmd\eta\; j\left(\sqrt{\left(\frac s\eps -\xi'\right)^2+|\eta|^2}\right) \bar m'(\xi')\; \frac{\eta_1^2}2 + o (\eps)\;.
\end{split}
\ee
Inserting this expansion in the approximated expression for $S_\eps^{(3)}(\varphi_\eps)$ we obtain, 
\[
\lim_{\eps\to 0} S_\eps^{(3)}(\varphi_\eps) = \int_0^T\!\rmd t\int_{\Gamma(t)} \!\rmd\s\;  A_{\bar Q}(\kappa_t)\;,
\]
where
\[
A_{\bar Q}(\kappa_t)  := \frac {\kappa_t^2}{2\beta} \int\!\rmd\xi\; \bar a \sqrt{1-\bar m^2} \; \left[\b\tilde J*(\bar m' \bar{Q})- \b(\tilde J*\bar m')\bar{Q} - \b f \right]^2\;,
\]
with
\be
\label{f}
f(\xi) =\int\!\rmd\xi'\int_{\bb R^{d-1}}\!\rmd\eta\; j\left(\sqrt{\left(\xi -\xi'\right)^2+|\eta|^2}\right) \bar m'(\xi')\; \frac{\eta_1^2}2\;.
\ee
Recalling the definitions \eqref{L}, \eqref{nu}, and using $\bar m' = (1-\bar m^2)\b \tilde J *\bar m'$, we get,
\[
A_{\bar Q}(\kappa_t) = \frac {\kappa_t^2}{4\beta} \int\!\nu(\rmd\xi)\; \left[L(\bar m' \bar{Q}) - H \right]^2\;,
\]
where
\begin{equation}
\label{H}
H :=  \b\, 2\bar a \sqrt{1-\bar m^2}\, f \;.
\end{equation}
By \eqref{N} and \eqref{nu},
\[
\frac{\int\!\nu(\rmd\xi)\; \bar m'(\xi)  H(\xi)}{\int\!\nu(\rmd\xi)\;(\bar m')^2} =  N\b \int \!\rmd\xi\;\bar m'(\xi)f(\xi) =  N\b\tau = \theta\;,
\]
where in the last equalities we used that, by \eqref{tau}, $\tau = \int \!\rmd\xi\;\bar m'(\xi)f(\xi)$, and the relations \eqref{mu} and \eqref{theta}. It follows that the component of $H$ orthogonal to $\bar m'$ in $L^2(\bb R,\nu(\rmd\xi))$ is 
\begin{equation}
\label{cH}
\hat H = H - \theta \bar m'\;.
\end{equation}
Therefore, by the symmetry of $L$ and $L\bar m'=0$, 
\[
A_{\bar Q}(\kappa_t) = \frac{(\theta\kappa_t)^2}{4 \mu}+ \frac 1{4\beta}  \int\!\nu(\rmd\xi)\; \left[L(\bar m' \bar Q) - \hat H \right]^2\;.
\]
The corrector $\bar Q$ is now determined by minimizing the above expression. More precisely, $\bar Q$ is the solution to the equation $L(\bar m' \bar Q) =\hat H$ which satisfies \eqref{Q1} and $\bar{Q}(0)=0$, whose existence and uniqueness is the content of Lemma \ref{lem:Q} in Appendix \ref{app:a}. Moreover, with this choice,
\be
\label{s4}
\lim_{\eps\to 0} S_\eps^{(3)}(\varphi_\eps) = \frac{1}{\mu} \int_0^T\!\rmd t\int_{\Gamma(t)} \!\rmd\s\;  \frac{(\theta\kappa_t)^2}4\;.
\ee
By \eqref{s1}, \eqref{s2}, and \eqref{s4} the statement (b) of the theorem follows.

\medskip
\noindent {\it Proof of \rm (a).}
By Legendre duality,
 $\mc L_\eps(u,v)= \sup_p \; \{pv-\mc H_\eps(u,p)\}$, where, given measurable functions $u\colon \bb T^d \to [-1, 1]$ and $\eta \colon \bb T^d \to \bb R$,
\begin{equation}
\label{Heps=}
\begin{split}
\mc H_\eps(u,\eta) & = \eps^{-2}\frac{c_\eps(u)}{\beta} \Big[\cosh(\beta J_\eps*u+2\beta \eta) -\cosh(\beta J_\eps*u) \\ & \qquad\qquad  - u\sinh(\beta J_\eps*u+2\beta \eta) + u\sinh(\beta J_\eps*u)\Big]\;.
\end{split}
\end{equation}
Whence, letting $\varphi_\eps$ be as in \eqref{recovery}, for each $g=g(t,x)$, 
\[
\begin{split}
S_\eps(\varphi_\eps) & \ge \eps^{-1}\int_0^T\!\rmd t\int \!\rmd x\; \bigg\{ \dot\varphi_\eps g - \eps^{-2}   \frac{c_\eps(\varphi_\eps)}{\beta} \Big[\cosh(\beta J_\eps*\varphi_\eps+2\beta g) \\ & \qquad  -\cosh(\beta J_\eps*\varphi_\eps)- \varphi_\eps\sinh(\beta J_\eps*\varphi_\eps+2\beta g) + \varphi_\eps\sinh(\beta J_\eps*\varphi_\eps)\Big]\bigg\} \\ & =: \Lambda_\eps(\varphi_\eps,g) \;.
\end{split}
\]
Given a fixed smooth function $p=p(t,x)$, we choose
\[
g(t,x) = g_\eps(t,x) = \eps N  p(t,x) \left[\frac{\bar m'(s/\eps)}{2\bar a(s/\eps)\sqrt{1-\bar m(s/\eps)^2}}  \right]_{s=d(x,\Gamma(t))}
\]
and compute the limit of $\Lambda_\eps(\varphi_\eps,g_\eps)$ as $\eps\to 0$. By second order Taylor expansion of $\mc H_\eps(u,\cdot)$ and observing the the remainder are equibounded and converges to zero point-wise as $\eps\to 0$, we have
\[
\Lambda_\eps(\varphi_\eps,g_\eps) = \Lambda_\eps^{(1)}(\varphi_\eps,g_\eps)  + \Lambda_\eps^{(2)}(\varphi_\eps,g_\eps)  + \Lambda_\eps^{(3)}(\varphi_\eps,g_\eps) + o(1) \;,
\]
where
\[
\Lambda_\eps^{(1)}(\varphi_\eps,g_\eps) = \int_0^T\!\rmd t\int \!\rmd x\; \eps^{-1} \dot\varphi_\eps g_\eps\;,
\]
\[
\begin{split}
\Lambda_\eps^{(2)}(\varphi_\eps,g_\eps) & = \int_0^T\!\rmd t\int \!\rmd x\; \eps^{-3} c_\eps(\varphi_\eps)\cosh(\beta J_\eps*\varphi_\eps) \big[ \varphi_\eps - \tanh(\beta J_\eps*\varphi_\eps) \big]2g_\eps \\ & =  \int_0^T\!\rmd t\int \!\rmd x\; \eps^{-3} c_\eps(\varphi_\eps) \sqrt{1-\varphi_\eps^2} \sinh\bigg(\eps\b\frac{\delta F^\epsilon}{\delta m}(\varphi_\eps)\bigg) 2g_\eps\;,
\end{split}
\]
\[
\Lambda_\eps^{(3)}(\varphi_\eps,g_\eps)  = - \int_0^T\!\rmd t\int \!\rmd x\; \eps^{-3}  c_\eps(\varphi_\eps) \cosh(\beta J_\eps*\varphi_\eps) \big[ 1-\varphi_\eps \tanh(\beta J_\eps*\varphi_\eps) \big] 2\beta g_\eps^2\;.
\]

By \eqref{Q} and the assumptions on $R_\eps$, the expansion \eqref{phid} holds with an extra additive $o(\eps)$ due to the presence of $R_\eps$. Therefore, as $g_\eps$ is equibounded and recalling \eqref{N}, the same reasoning leading to \eqref{s2} gives,
\[
\begin{split}
\lim_{\eps\to 0}\Lambda_\eps^{(1)}(\varphi_\eps,g_\eps) & = \lim_{\eps\to 0}\int_0^T\!\rmd t\int\limits_{|s|\le C\eps(\log\eps)^2}\!\rmd s\int_{d=s}\!\rmd\s\; \eps^{-1}\;\frac{-\bar m'(s/\eps)^2Np\,\partial_td }{2\bar a(s/\eps)\sqrt{1-\bar m(s/\eps)^2}} \\ & =- \int_0^T\!\rmd t\int_{\Gamma(t)} \!\rmd\s\; v_t p\;.
\end{split}
\]
Concerning $\Lambda_\eps^{(2)}$ and $\Lambda_\eps^{(3)}$, we observe that, as $g_\eps=O(\eps) \bar m'(d/\eps)$ and the dependence on $\varphi_\eps$  of the integrands is locally Lipschitz, the contribution due to $R_\eps$ is $o(1)$ as $\eps\to 0$, and therefore can be neglected. 

Noticing that \eqref{edf} holds true here with $Q$ in place of $\mc K\bar Q$ and recalling \eqref{f} we have,
\[
\begin{split}
\lim_{\eps\to 0}\Lambda_\eps^{(2)}(\varphi_\eps,g_\eps) & = \lim_{\eps\to 0} \int_0^T\!\rmd t\int\limits_{|s|\le C\eps(\log\eps)^2}\!\rmd s\int_{d=s}\!\rmd\s\; \;N p \\ & \quad\qquad \times \eps^{-1}\Big\{\bar m' \big[-\b\tilde J*(\bar m' Q)+ \b(\tilde J*\bar m')Q+ \b \mc K f \big]\Big\}(s/\eps) \\ & = \int_0^T\!\rmd t\int_{\Gamma(t)} \!\rmd\s\; Np \int\!\nu(\rmd\xi)\; \bar m' \big[\b\kappa_t\, 2\bar  a \sqrt{1-\bar m^2}\, f -L(\bar m' Q)\big] \\ & = \int_0^T\!\rmd t\int_{\Gamma(t)} \!\rmd\s\;\theta\kappa_t p\;,
\end{split}
\]
where in the last identity we used that $\int\!\nu(\rmd\xi)\; \bar m' L(\bar m' Q) = 0$. Finally, as $\cosh(\beta J_\eps*\varphi_\eps) \big[ 1-\varphi_\eps \tanh(\beta J_\eps*\varphi_\eps) \big] = \sqrt{1-\bar m(s/\eps)^2} + o(1)$, 
\[
\begin{split}
\lim_{\eps\to 0}\Lambda_\eps^{(3)}(\varphi_\eps,g_\eps) & = \lim_{\eps\to 0} \int_0^T\!\rmd t\int_{\bb R}\!\rmd s\int_{d=s}\!\rmd\s\; \; \eps^{-1}\beta N^2 p^2 \frac{-\bar m'(s/\eps)^2}{2\bar a(s/\eps)\sqrt{1-\bar m(s/\eps)^2}} \\ & \xrightarrow{\eps\to 0} \int_0^T\!\rmd t\int_{\Gamma(t)} \!\rmd\s\; (-\mu p^2) \;.
\end{split}
\]

We conclude that, for any function $p=p(x,t)$, 
\[
\varliminf_{\eps\to 0} S_\eps(\varphi_\eps)  \ge \int_0^T\!\rmd t\int_{\Gamma(t)} \!\rmd\s\;(-v_tp+\theta\kappa_t p - \mu p^2)\;,
\]
whence, by optimizing over $p$,
\[
\varliminf_{\eps\to 0} S_\eps(\varphi_\eps)  \ge \int_0^T\!\rmd t\int_{\Gamma(t)} \!\rmd\s\;\sup_p\;  (-v_tp+\theta\kappa_t p - \mu p^2) = \frac 1{4\mu}\int_0^T\!\rmd t\int_{\Gamma(t)}\!\rmd\s\; (v_t-\theta \kappa_t)^2 \;.
\]
The statement (a) of the theorem is thus proved.
\qed      
                            
\section{Glauber+Kawasaki process}
\label{sec:4}

In this section we analyze the sharp interface limit of the action functional in the context of the Glauber+Kawasaki process.

\subsection{Motivation}
\label{sec:4.1}

The so-called Glauber+Kawasaki process is a simple stochastic model describing a chemical reaction among two species together with their diffusion. Recall that $\bb T_{L}^d$ denotes the $d$-dimensional torus of side $L$ and, given an integer $N\ge 1$, in this section we let $\bb T^d_{L,N} := (LN^{-1} \bb Z / L \bb Z)^d $ be the discrete approximation of $\bb T^d_L$ with lattice spacing $L/N$. Set also $\Omega_{N,L}:=\{0,1\}^{\bb T_{L,N}^d}$, if $\eta\in\Omega_{N,L}$ we regard its value at the site $i\in \bb T_{L,N}^d$, that can be either zero or one, as representing the species occupying $i$. The Glauber+Kawasaki process is a continuous time Markov chain on the state space $\Omega_{N,L}$, whose dynamics is obtained superimposing two elementary mechanisms, respectively modeling the reaction (Glauber) and the
diffusion (Kawasaki). Namely, the generator of the chain is 
\begin{equation}
	\label{g+k}
	\mc L_{N} :=  \mc L_{\mathrm{G}} + N^2 \mc L_{\mathrm{K}}\;. 
\end{equation}
Let $c$, a strictly positive local function of the configuration, be the rate of the reaction; then, given $f\colon \Omega_{N,L}\to \bb R$, 
\begin{equation*}
	\mc L_{\mathrm{G}} f \, (\eta) := 
	\sum_{i\in \bb T_{L,N}^d} 
	c(\tau_i \eta) \big[ f(\eta^i) - f(\eta)\big] \;,
\end{equation*}
where $\tau_i$ is the translation, i.e., $(\tau_i\eta)_j:= \eta_{j-i}$, and $\eta^i$ is the configuration obtained from $\eta$ by flipping the occupation number at $i$. The Kawasaki dynamics is instead defined by the generator,
\begin{equation*}
	\mc L_{\mathrm{K}} f \, (\eta) :=	\frac 12 \sum_{\{i,j\}}  \big[ f (\eta^{i,j}) - f(\eta) \big]\;,
\end{equation*}
where the sum runs over the (unordered) nearest neighbors pairs $\{i,j\} \subset \bb T_{L,N}^d$ and $\eta^{i,j}$ is the configuration obtained from $\eta$ by exchanging the occupation numbers at the sites $i$ and $j$. Note that in \eqref{g+k} the Kawasaki dynamics has been speeded up by $N^2$; as the lattice spacing is $L/N$ this corresponds to a diffusive rescaling. Let $\mc M_{+}(\bb T_L^d)$ be the set of positive measures on $\bb T_L^d$ and define the \emph{empirical density} as the map $\pi^N\colon \Omega_{L,N} \to \mc M_{+}(\bb T_L^d)$ given by
\begin{equation*}
	\pi^N (\eta) = \frac 1{N^d} \sum_{i\in \bb T_{L,N}^d} \eta_i \delta_{i}\;.
\end{equation*}
Assuming that the initial datum  $\eta(0)$ for the Glauber+Kawasaki process is well prepared, in the sense that $\pi^N (\eta(0)) \to u_0(x) \rmd x$ for some Borel function $u_0\colon \bb T_L \to [0,1]$, in \cite{DFL} it is proven that $\pi^N (\eta(t)) \to u(t,x) \rmd x$ in probability, where $u\colon [0,\infty)\times \bb T_L^d \to [0,1]$ solves the reaction diffusion equation,
\begin{equation}
	\label{hy}
	\begin{cases}
		\partial_t u = \frac 12 \Delta u + B(u) - D(u)\;,\\
		u(0) =u_0\;.
	\end{cases}
\end{equation}
The reaction term is described by the coefficients $B,D\colon [0,1] \to [0,+\infty)$ that can be obtained from the microscopic rate $c\colon \Omega_{L,N} \to (0,\infty)$ according to the following procedure. For $\rho\in [0,1]$ let $\nu_\rho$ be the Bernoulli measure with parameter $\rho$, namely the product probability on $\Omega_{L,N}$ with marginals $\nu_\rho(\eta_i=1) = \rho$. Then, 
\begin{equation*}
	B(\rho) = \nu_\rho\big( (1-\eta_0) c \big)\;,
	\qquad\qquad D(\rho) = \nu_\rho\big( \eta_0 c \big)\;.
\end{equation*}
Observe that, as $c$ is a strictly positive local function, $B$ and $D$ are strictly positive polynomials in $(0,1)$, while $B(1)=0$ and $D(0)=0$.  

The hydrodynamic equation \eqref{hy} describes the typical behavior of the Glau\-ber+Kawasaki process in the diffusive scaling limit. On the other hand, the statistics of the fluctuations cannot be described simply by adding a Gaussian noise to \eqref{hy}. In fact, as precised by the large deviation theory, the Poissonian nature of the underlying Glauber dynamics is still felt in the diffusive limit.  A main motivation for analyzing the large deviations properties of the empirical density is the following. Since the Glauber+Kawasaki process is irreducible, by general criterion on Markov chains, there exists a unique stationary probability $\mu_{L,N}$ on $\Omega_{L,N}$. As the dynamics does not satisfy the detailed balance condition, $\mu_{L,N}$ cannot be written in a closed form (with the exception of the special choices discussed in \cite{GJLV}) and, as shown in \cite{BJ}, it exhibits long range correlations. According to the general ideology on thermodynamics limits, we are not really interested in the whole details of the probability $\mu_{L,N}$, but mainly in the statistics of the empirical density in the limit $N\to \infty$.  It is therefore natural to introduce the sequence of probabilities $\{\wp_{L,N}\}_{N\in \bb N}$ on $\mc M_+(\bb T^d_L)$ defined by $\wp_{L,N} := \mu_{L,N} \circ (\pi^N)^{-1}$ and look for its asymptotic behavior as $N\to \infty$. Let $W\colon [0,1] \to \bb R$ be such that $B-D =-W'$. If $W$ has a unique minimizer, it is natural to expect that the sequence $\{\wp_{L,N}\}_{N\in \bb N}$ converges to the stationary solution of \eqref{hy} corresponding to the minimizer of $W$. Indeed, in the one dimensional case, this is proven in \cite{BPSV1} when $W$ has a single well, and in \cite{BPSV2} when $W$ has a double well.

A finer description of the asymptotics of $\{\wp_{L,N}\}_{N\in \bb N}$ can be achieved by looking at its large deviations. In a sloppy notation, this means the estimate 
\begin{equation*}
	\wp_{L,N} \big( \pi \sim u\, \rmd x \big) \asymp \exp\{ -N^d F_L(u) \}\;,
\end{equation*}
for a suitable functional $F_L$ on the set of densities $u\colon \bb T_L^d\to [0,1]$. Here $F_L$ plays the same role as the Cahn-Hilliard functional in the gradient theory of phase transition or the Lebowitz-Penrose functional \eqref{F}, with the minor inconvenient that it is not known.

According to the Friedlin-Wentzel theory for diffusions on $\bb R^n$, see \cite{FW}, the functional $F_L$ can be characterized in terms of a dynamical problem. To this end, fix $T>0$, a sequence of initial configurations $\eta^N(0)$, and consider the large deviations asymptotics for the empirical measure in the time window $[0,T]$. Under the assumption that $B$ and $D$ are concave, this large deviation principle has been proved in \cite{JLV,BL,LT} in one dimension (however, the result can be extended to higher dimension), the corresponding rate function, denoted by $I_{T,L}$, will be recalled later.  As proven in \cite[Chp.~6]{FW} for diffusions on $\bb R^n$ and in \cite{FLT} for the present setting (in one dimension and with additional hypotheses on the coefficients $B$ and $D$ implying a complete characterization of the stationary solutions to \eqref{hy}), the functional $F_L$ is the \emph{quasi-potential} associated to the dynamical rate functional $I_{T,L}$. This means that $F_L$ can be obtained from  $I_{T,L}$ by solving a suitable variational/combinatorial problem whose details are here omitted. Within this approach, in \cite{FLT} it is deduced that the cluster points of  $\{\wp_{L,N}\}$ are supported by the stationary solutions to \eqref{hy} associated to the minimizers of $W$.

We consider here the case of a bistable reaction term. Recalling $W$ satisfies $B-D =-W'$, we thus assume that $W$ has a twofold degenerate quadratic minimum, namely, there exist $0<\rho_-<\rho_+<1$ such that for $\rho\neq \rho_\pm$ we have $W(\rho) > W(\rho_-)=W(\rho_+)$ and $W''(\rho_-),W''(\rho_+) >0$.  In this situation, the probability $\mu_{L,N}$ describes the phase coexistence of the two stable phases, like a Gibbs measure undergoing a first order phase transition. Our purpose is to characterize the corresponding surface tension, that measures the cost of a transition between the two stable phases $\rho_\pm$.

As in the case of Ising-Kac model, the surface tension is identified by considering the sharp interface limit. By setting $\eps=L^{-1}$, as far as the dynamical behavior is concerned, under diffusive rescaling, the joint limit $N\to\infty$ and $\eps\to 0$ (with $\eps\gg  N^{-1}$) of the empirical measure has been analyzed in \cite{Bo,KS0}. More precisely, it is there proven that the limiting dynamics is described by the motion by mean curvature of the interface separating the stable phases, respectively in the classical setting \cite{Bo}  and in the level set formulation \cite{KS0}. 

In order to analyze the asymptotic behavior of the probability $	\wp_{L,N}$, let us introduce the family of functionals $F^\epsilon$ on the set of densities $u\colon \bb T^d \to [0,1]$ defined by 
\begin{equation*}
F^\epsilon (u) := \epsilon^{d-1} F_{\epsilon^{-1}} \big( u (\epsilon \cdot) \big)\;.
\end{equation*}
We are next going to argue, but not rigorously prove, that as $\epsilon\to 0$ the sequence $ F^\epsilon$ converges to a functional $F$ that is finite only if $u\in BV \big(\bb T^d;\{\rho_-,\rho_+\}\big)$ and for such $u$ is proportional to the (measure theoretic) perimeter of the jump set of $u$. Namely, $F(u)=\tau \mc H^{d-1}(S_u)$, where $ \mc H^{d-1}$ is the $(d-1)$-dimensional Hausdorff measure on $\bb T^d$ and $S_u$ denotes the jump set of $u$. The constant $\tau>0$ is then identified with the surface tension for the Glauber+Kawasaki processes and it will be characterized in terms of the solution to a one-dimensional ODE. 

As the quasi-potential $F_L$ is not directly accessible, we shall consider the sharp interface limit of the dynamical rate function $I_{T,L}$. More precisely, let $S_\epsilon$ be the functional on the set of paths $\phi\colon [0,T]\times \bb T^d \to [0,1]$ defined by $S_\epsilon (\phi) := \epsilon^{d-1} I_{\epsilon^{-2}T, \epsilon^{-1}} \big( \phi(\epsilon^2\cdot, \epsilon \cdot)\big)$. In Theorem~\ref{thm:4.1} below we prove that, for suitable sequences $\phi_\epsilon$ converging to
\begin{equation*}
	\phi(t,x)= 
	\begin{cases}
		\rho_+ & \textrm{ if } x\in \Omega(t)\;,\\
		\rho_- & \textrm{ if } x\not\in \bar \Omega(t)\;,\\
	\end{cases}
\end{equation*}
for some open $\Omega(t) \subset \bb T^d$ with smooth boundary, 
\begin{equation}
	\label{ldrf}
	\lim_{\epsilon\to 0} S_\epsilon(\phi_\epsilon) =\frac 1{2} \tau \int_0^T\!\rmd t\int_{\Gamma(t)}\!\rmd\s\;
	\left(v_t- \frac 12 \kappa_t\right)^2\;,
\end{equation}
where $\Gamma(t) =\partial \Omega(t)$, $\rmd\sigma$ is the surface measure on $\Gamma(t)$, $v_t$ is the normal velocity of $\Gamma(t)$, $\kappa_t$ is its mean curvature, and $\tau$ is a positive constant.

Observe now that the limiting dynamical rate function in \eqref{ldrf} measures, in $L^2$ sense, the deviations with respect to the motion by mean curvature $v_t = \frac 12 \kappa_t$. Since this evolution is - informally - the gradient flow of (one half of) the perimeter, we deduce that the quasi-potential associated to the limiting dynamical rate function is proportional to the perimeter and we identify the proportionality constant with $\tau$, see \cite[Thm.~4.3.1]{FW} for a proof of this statement in the context of diffusions in $\bb R^n$.

\subsection{Preliminaries}
\label{sec:4.2}

Let $\bar u$ be the instanton  (standing wave) associated to the hydrodynamic equation \eqref{hy} in dimension one, namely the solution to
\begin{equation}
	\label{inst1}
	\frac 12 \bar u '' + B(\bar u) - D(\bar u) = 0\;, \quad \bar u(\pm\infty) = \rho_\pm\;, \quad \bar u(0) = \frac{\rho_++\rho_-}2\;.
\end{equation}
Clearly $\bar u'(\xi)>0$ and it can be easily shown that  
\begin{equation}
\label{asbaru}
\sup_{\xi\in\bb R} \big(\big|\bar u'(\x)\big| + |\bar u''(\xi)\big|\big) \rme^{\gamma|\xi|}  <+\infty\;,
\end{equation}
where $\gamma=\min\{D'(\rho_+)-B'(\rho_+);D'(\rho_-)-B'(\rho_-)\}$.

The large deviation asymptotics for the empirical density under the Glauber+Ka\-wasaki  dynamics has been analyzed in \cite{BL,JLV}. We next recall the associated rate function. Given $L$ positive, let $C(L):=\{\rho\in L^\infty(\bb T_L^d) \colon 0\le \rho \le 1\}$, where $\bb T_L^d$ is the $d$-dimensional torus of side $L>0$, equipped with the (metrizable) weak*-topology. We define $I_{T,L} \colon C([0,T];C(L)) \to [0, \infty]$ by
\begin{equation}
\label{Igk}
I_{T,L}(\phi) := \sup_{H\in C^{1,2}([0,T]\times\bb T_L^d)} J^H_{T,L}(\phi) \;,
\end{equation}
where   
\begin{equation}
\label{Jgk}
\begin{split}
&J^H_{T,L}(\phi) := \int\!\rmd r\, \big[\phi(T,\cdot)H(T,\cdot) - \phi(0,\cdot)H(0,\cdot)\big] \\ & \quad - \int_0^T\!\rmd t \int\!\rmd r\, \bigg[\phi\bigg(\partial_t H + \frac 12 \Delta H \bigg) + \frac 12 \phi (1-\phi) |\nabla H |^2 \bigg] \\ & \quad -  \int_0^T\!\rmd t \int\!\rmd r\, \bigg[B(\phi) (\rme^H-1) + D(\phi)(\rme^{-H}-1)\bigg]\;.
\end{split}
\end{equation}
Under suitable assumptions on the initial conditions, in\cite{BL,JLV,LT} it is proven that the empirical magnetization sampled according to the  Glauber dynamics, regarded as a random variable taking values in the Skorokhod space $D([0,T];\mc M(\bb T^d_L))$, satisfies a large deviation principle with speed $N^d$ and rate function $\mc I_{T,L}$ given by $\mc I_{T,L}(\nu) = I_{T,L}(\phi)$ if $\nu_t= \phi_t \, \rmd r$ for some $\phi \in C([0,T];C(L))$ and $+\infty$ otherwise. 

By \cite[Lemma 2.1]{JLV}, or rather its generalization in dimension $d\ge 1$, if $\phi\in C^{2,3}([0,T]\times \bb T_L^d;(0,1))$ then the supremum in \eqref{Igk} is achieved for $H=H(\phi) \in C^{1,2}([0,T]\times \bb T_L^d)$, the unique classical solution to the non-linear Poisson equation, 
\begin{equation}
\label{Ihke}
\partial_t \phi + \nabla\cdot [\phi(1-\phi)\nabla H]  = \frac 12 \Delta \phi + B(\phi)\rme^H - D(\phi)\rme^{-H}\;,
\end{equation}
so that, for such $H$,
\begin{equation}
\label{Ihk1}
\begin{split}
& I_{T,L}(\phi)  = J^H_{T,L}(\phi) = \frac 12 \int_0^T\!\rmd t \int \!\rmd r\, \phi(1-\phi) |\nabla H|^2  \\ &  +  \int_0^T\!\rmd t \int\!\rmd r\, B(\phi) \big(1-\rme^H + H \rme^H\big) +  \int_0^T\!\rmd t \int\!\rmd r\, D(\phi)  \big(1-\rme^{-H} - H \rme^{-H}\big)\;.
\end{split}
\end{equation}
Due to the lack of reversibility of the underlying microscopic dynamics, it is not possible to decompose the action function in a form analogous to \eqref{gg*}. 

\subsection{Sharp interface limit of the action functional} 
\label{sec:4.3}

To this end, we set $\eps=L^{-1}$, perform a diffusive rescaling of space and time, and normalize the resulting action with a factor $L^{d-1}$. As in the previous section, the space variable in $\bb T^d$ is denoted by $x$. We then introduce the rescaled action functional renormalized with a factor $L^{d-1}$. We thus define the rescaled functional $S_\eps \colon C([0,T];C(1)) \to [0,\infty]$ by
\be
\label{Sgk}
S_\eps(\phi) = \eps^{d-1} I_{\eps^{-2}T,\eps^{-1}}(\phi(\eps^2\cdot, \eps \cdot )) \;,
\ee
whose variational representation is 
\be
\label{Sgk1}
S_\eps(\phi) := \sup_{H\in C^{1,2}([0,T]\times\bb T^d)} J^H_\eps(\phi) \;,
\ee
where   
\begin{equation}
\label{Jgk1}
\begin{split}
&J^H_\eps(\phi) := \frac 1\eps\int\!\rmd x\, \big[\phi(T,\cdot)H(T,\cdot) - \phi(0,\cdot)H(0,\cdot)\big] \\ & \quad - \frac 1\eps \int_0^T\!\rmd t \int\!\rmd x\, \bigg[\phi \bigg(\partial_t H + \frac 12 \Delta H \bigg) + \frac 12 \phi(1-\phi) |\nabla H|^2 \bigg]  \\ & \quad -  \frac 1\eps \int_0^T\!\rmd t \int\!\rmd x\, \bigg(B(\phi) \frac{\rme^H-1}{\eps^2} + D(\phi) \frac{\rme^{-H}-1}{\eps^2}\bigg)\;.
\end{split}
\end{equation}
Moreover,  the representation \eqref{Ihke}, \eqref{Ihk1} gives, for $\phi\in C^{2,3}([0,T]\times \bb T^d;(0,1))$,
\begin{equation}
\label{Ihke1}
\partial_t \phi + \nabla\cdot [\phi(1-\phi)\nabla H]  = \frac 12 \Delta \phi + \frac{B(\phi)\rme^H - D(\phi)\rme^{-H}}{\eps^2}
\end{equation}
and 
\begin{equation}
\label{Ihk2}
\begin{split}
S_\eps(\phi) & = \frac 1{2\eps} \int_0^T\!\rmd t \int \!\rmd x\, \phi(1-\phi) |\nabla H|^2  \\ & \quad + \frac 1{\eps^3}  \int_0^T\!\rmd t \int\!\rmd x\, B(\phi) \big(1-\rme^H + H \rme^H\big) \\ & \quad +  \frac 1{\eps^3} \int_0^T\!\rmd t \int\!\rmd x\, D(\phi)  \big(1-\rme^{-H} - H \rme^{-H}\big)\;.
\end{split}
\end{equation}

As in the previous section, given a $C^1$-family of oriented smooth surfaces $\Gamma=\{ \Gamma(t)\}_{t\in [0,T]}$, with $\Gamma(t) = \partial\Omega(t)$ for some open  $\Omega(t) \subset \bb T^d$, we denote by $n_t=n_{\Gamma(t)}$ the inward normal of $\Gamma(t)$, by $v_t\colon \Gamma(t) \to \bb R$ the normal velocity of $\Gamma$ at time $t$, by $\kappa_t$ the mean curvature of $\Gamma(t)$, and  by  $d(\cdot,\Gamma(t))$ a regularized version of the signed distance from $\Gamma(t)$. 

For such families of surfaces we define the limiting action functional,
\be
\label{Rpn1}
S_\mathrm{ac}(\Gamma) = \frac 1{4\mu}\int_0^T\!\rmd t\int_{\Gamma(t)}\!\rmd\s\; \left(v_t- \frac 12 \kappa_t\right)^2\;.
\ee
where the mobility $\mu$ is computed according to the following procedure. Recalling the definition \eqref{inst1}, let $L_{\bar u}$ be the linear operator given by
\begin{equation}
\label{Lbar}
L_{\bar u} \psi = \big[(\bar u (1-\bar u) \psi'\big]' - [B(\bar u) + D(\bar u)]\psi\;, 	
\end{equation}
which is obtained by linearizing \eqref{Ihke} in dimension one at $\phi=\bar u$ around $H=0$. Then,
\begin{equation}
\label{mu1}
\mu =\frac{2\langle\bar u', (-L_{\bar u})\bar u'\rangle_{L^2} }{ \|\bar u'\|_{L^2}^4}\;.
\end{equation}
For later purpose, we notice that since $B+D$ is strictly positive then $L_{\bar u}$ is bijective on $L^2(\bb R)$. Moreover, the inverse of $L_{\bar u}$ preserves the decays properties, in the sense that if $L_{\bar u}\psi = w$ then, for any $\gamma'>0$,
\begin{equation}
\label{decayl}
\sup_{\xi\in\bb R}|w(\xi)|\rme^{\gamma'|\xi|} < + \infty  \quad \Longrightarrow \quad \sup_{\xi\in\bb R} \big(|\psi(\xi)| + |\psi'(\xi)| + |\psi''(\xi)|\big) \rme^{\gamma'|\xi|}< +\infty\;.
\end{equation}
\begin{theorem}
\label{thm:4.1}
Given a $C^1$-family of oriented smooth surfaces $\Gamma=\{ \Gamma(t)\}_{t\in [0,T]}$, with $\Gamma(t) = \partial\Omega(t)$, consider sequences $\{ \phi_\eps \} \subset C([0,T];C(1))$, converging to $\rho_- + (\rho_+ - \rho_-) \id_{\Omega(\cdot)}$, of the form  
\begin{equation}
\label{recovery1}
\phi_\eps(t,x) = \bar u\left(\frac{d(x,\Gamma(t))}\eps + \eps Q\left(t,x,\frac{d(x,\Gamma(t))}\eps\right)\right)+\eps R_\eps(t,x)\;,
\end{equation}
where $\bar u$ is the instanton, $Q \colon [0,T] \times \bb T^d \times \bb R \to \bb R$ is a smooth function such that
\be
\label{Q3}
\begin{split}
& \sup_{(t,x,\xi) \in [0,T] \times \bb T^d \times \bb R} \bigg\{ \frac{\big| Q(t,x,\xi) \big| + \big| \partial_\xi Q(t,x,\xi)\big|}{1+|\xi|} \\ & \quad\qquad\qquad + \frac{\big|\partial_t Q(t,x,\xi) \big| +\big|D_x Q(t,x,\xi) \big| + \big|D^2_{xx}Q(t,x,\xi)\big|}{1+|\xi|}\bigg\} < + \infty\;,
\end{split}
\ee
and $R_\eps \colon [0,T] \times \bb T^d \to \bb R$ is a smooth function.
\begin{itemize}
\item[(a)] If $\| R_\eps\|_\infty +\| \partial_t R_\eps \|_\infty +\| \Delta R_\eps \|_\infty \to 0$ as $\eps \to 0$  then, for any $Q$,
\[
\liminf_{\eps\to 0} S_\eps(\phi_\eps) \ge S_\mathrm{ac}(\Gamma)\;.
\]
\item[(b)] There exist $Q^*$ such that, choosing $Q=Q^*$ and $R_\eps=0$ we have,
\[
\lim_{\eps\to 0} S_\eps(\phi_\eps) = S_\mathrm{ac}(\Gamma)\;.
\]
\end{itemize}
\end{theorem}

For expository reasons, we prove the statements in reverse order.

\noindent {\it Proof of \rm (b).} In the sequel, we assume $R_\eps =0$ and $Q(t,x,\xi) := A(t,x)\bar Q(\xi)$, where $A\colon [0,T]\times\bb T^d \to \bb R$ and $\bar Q\colon \bb R\to \bb R$ are smooth functions to be determined later, with $\bar Q$ such that 
\be
\label{Q2}
\sup_{\xi\in\bb R} \bigg\{ \frac{\big| \bar Q(\xi) \big|}{1+|\xi|} + \big|\bar Q'(\xi)\big|  + \big|\bar Q''(\xi)\big|\bigg\}   < + \infty\;.
\ee
In order to compute the cost of the sequence \eqref{recovery1} with these choices, we start by the expansions,
\begin{equation}
\label{expsu}
\begin{split}
\phi_\eps  & = \bar u(d_\eps) + \eps\bar u'(d_\eps) Q_\eps + \eps R^{(1)}_\eps \;, \qquad \partial_t\phi_\eps = \bar u'(d_\eps) \frac{\partial_t d}\eps + R^{(2)}_\eps \;, \\ \Delta\phi_\eps  & = \frac{\bar u''(d_\eps)}{\eps^2} +\frac {\bar u'(d_\eps) \Delta d + \bar u'''(d_\eps) Q_\eps+  2 \bar u''(d_\eps) Q_\eps' + \bar u'(d_\eps) Q_\eps''}\eps+ R^{(3)}_\eps \; , 
\end{split}
\end{equation}
where the notation $d=d(x,\Gamma(t))$, $d_\eps = d/\eps$, $Q_\eps = Q(t,x,d_\eps)$, $Q_\eps'=\partial_\xi Q(t,x,d_\eps)$, and $Q_\eps''=\partial_{\xi\xi}^2 Q(t,x,d_\eps)$ has been adopted, and $R^{(i)}_\eps=R^{(i)}_\eps(t,x,d_\eps)$, $i=1,2,3$, are such that
\[
 \limsup_{\eps \to 0}	\sup_{(t,x,\xi) \in [0,T] \times \bb T^d \times \bb R} \rme^{\gamma |\xi|/2} |R^{(i)}_\eps(t,x,\xi)|  < \infty\;,
\]
with $\gamma$ as in \eqref{asbaru}. Eqs.\eqref{expsu} can be easily derived using \eqref{asbaru} and recalling that $|\nabla d|=1$ in a neighborhood of $\Gamma(t)$.

Next, assuming for the function $H$, unique solution to \eqref{Ihke1}, an expansion of the form $H(t,x) = \eps H_1(t,x,d(x,\Gamma(t))/ \eps)) + O(\eps^2)$, we deduce a linear equation for $H_1$.  To this end, we write 
\begin{equation}
\label{expsu1}
\begin{split}
 & \frac{B(\phi_\eps)\rme^H - D(\phi_\eps)\rme^{-H}}{\eps^2} = \frac{B(\bar u(d_\eps)) - D(\bar u(d_\eps))  }{\eps^2}  \\ & \quad +  \frac{B'(\bar u(d_\eps)) - D'(\bar u(d_\eps)) }{\eps} \bar u'(d_\eps) Q_\eps  +  \frac{B(\bar u(d_\eps)) + D(\bar u(d_\eps)) }{\eps} H_1 + O(1)\;.
\end{split}
\end{equation}
Plugging \eqref{expsu} and \eqref{expsu1} in \eqref{Ihke1} and making use of \eqref{inst1} and its derivative, we deduce, after some straightforward computations, that, for $(t,x)\in [0,T]\times\bb T^d$ fixed, $H_1(t,x,\cdot)$ satisfies 
\begin{equation}
\label{H1}
(\bar u(1-\bar u) H_1')' - [B(\bar u) + D(\bar u)] H_1 =  \left(\frac 12 \Delta d -\partial_td\right) \bar u' + \bar u'' A\bar Q' + \frac 12 \bar u' A\bar Q''\;.
\end{equation}
Hence, by choosing $A = \partial_td - \frac 12 \Delta d$ and recalling \eqref{Lbar}, we get $H_1(t,x,\xi) = A(t,x) h(\xi)$ where $h\colon \bb R\to\bb R$ solves
\begin{equation}
\label{key}
L_{\bar u} h = - \bar u' + \bar u'' \bar Q' + \frac 12 \bar u' \bar Q''\;.
\end{equation}
For later purpose, we remark that, in view of \eqref{asbaru}, \eqref{Q2}, and \eqref{decayl},
\begin{equation}
\label{asbarh}
\sup_{\xi\in\bb R} \big(|h(\xi)| + |h'(\x)| + |h''(\xi)|\big) \rme^{\gamma|\xi|}  <+\infty\;.
\end{equation}
With this choice of $H_1$,  the initial assumption on the  expansion for $H$ holds. More precisely, in Appendix \ref{app:a} it is proven that $H(t,x) = \eps H_1(t,x,d(x,\Gamma(t))/ \eps)) + \eps^2 \tilde{H}_\eps(t,x)$ with
\begin{equation}
\label{Htilde}
\limsup_{\eps \to 0} \sup_{(t,x) \in [0,T] \times \bb T^d }  \big(  |\tilde{H}_\eps(t,x)|+\eps |\nabla \tilde{H}_\eps(t,x)|\big) <\infty \, .
\end{equation}

From \eqref{Ihk2}, the explicit form of $H_1$, and \eqref{Htilde} we then have,
\begin{equation}
\label{Ihk3}
S_\eps(\phi_\eps) = S_\eps^{(1)} + S_\eps^{(2)} + O(\eps) \;,
\end{equation}
where, after an integration by part,
\[
S_\eps^{(1)} = \frac 1{2\eps} \int_0^T\!\rmd t \int \!\rmd x\, \bigg[-\nabla\cdot (\phi_\eps (1-\phi_\eps)\, \eps^2 \nabla H_1) + \frac{B(\phi_\eps) + D(\phi_\eps)}2 H_1\bigg] H_1
\]
and 
\[
\begin{split}
S_\eps^{(2)}& =\frac 12 \int_0^T\!\rmd t \int \!\rmd x\, \phi_\eps (1-\phi_\eps)\, \eps^2 \nabla H_1 \cdot \nabla \tilde H_\eps \\ & \quad + \int_0^T\!\rmd t \int\!\rmd x\, \bigg[ (B(\phi_\eps)+D(\phi_\eps)) H_1\tilde H_\eps + \frac{B(\phi_\eps)-D(\phi_\eps)}3 H_1^3 \bigg] \;.
\end{split}
\]
Now, from \eqref{expsu}, \eqref{asbarh}, and \eqref{H1} we deduce,
\begin{equation}
\label{nuova}
\begin{split}
\nabla \cdot (\phi_\eps & (1-\phi_\eps) \nabla (\eps H_1))  - \frac{B(\phi_\eps)+D(\phi_\eps)}{\eps^2} (\eps H_1) \\  & = \frac 1\eps (\bar u(d_\eps)(1-\bar u(d_\eps)) H_1')' - \frac{B(\bar u(d_\eps)) + D(\bar u(d_\eps))}\eps H_1 + R^{(4)}_\eps \\ & =  \frac 1\eps \bigg[ \left(\frac 12 \Delta d -\partial_td\right) \bar u'(d_\eps) + \bar u''(d_\eps) A\bar Q' + \frac 12 \bar u'(d_\eps) A\bar Q''\bigg] + R^{(4)}_\eps\;,
\end{split}
\end{equation}
with $R^{(4)} = R^{(4)}_\eps(t,x,d_\eps)$, such that
\[
\limsup_{\eps \to 0}	\sup_{(t,x,\xi) \in [0,T] \times \bb T^d \times \bb R} \rme^{\gamma |\xi|/2} |R^{(4)}_\eps(t,x,\xi)|  < \infty\;.
\]
Recalling that $A(t,x) = \partial_td(x,\Gamma(t)) - \frac 12 \Delta d(x,\Gamma(t)) = v_t(x) - \frac 12 \kappa_t(x)$ for $x\in\Gamma(t)$ and using \eqref{asbaru}, \eqref{expsu}, \eqref{key}, and  \eqref{Htilde}, by applying the co-area formula as done in the proof of Theorem \ref{thm:3.1}, we can compute the limit of $S_\eps^{(1)}$ and $S_\eps^{(2)}$ as $\eps\to 0$. By few direct calculations (that we omit) we obtain,
\begin{equation}
\label{Ihk4}
\lim_{\eps\to 0} S_\eps^{(1)}= C_{\bar Q} \int_0^T\!\rmd t\int_{\Gamma(t)}\!\rmd\s\; \bigg(v- \frac 12 \kappa_t\bigg)^2\;, \qquad \lim_{\eps\to 0} S_\eps^{(2)} = 0\;, 
\end{equation}
where, denoting by $\langle\cdot,\cdot\rangle_{L^2}$ the scalar product in $L^2(\bb R;\rmd\xi)$,
\[
C_{\bar Q} = \frac 12 \Big\langle \Big(\bar u' - \bar u''\bar Q'-\frac 12 \bar u' \bar Q''\Big), (- L_{\bar u})^{-1}  \Big(\bar u' - \bar u''\bar Q'-\frac 12 \bar u' \bar Q''\Big) \Big\rangle_{L^2}\;.
\]

We observe that, since $\bar u''\bar Q'+\frac 12 \bar u' \bar Q'' \in L^2(\bb R;\rmd\xi)$ and $\big\langle\bar u', \bar u''\bar Q'+\frac 12 \bar u' \bar Q''\big\rangle_{L^2} = 0$,
\begin{equation}
\label{C*}
C_{\bar Q} \ge C^* := \frac 12 \min_{\psi:\langle \bar u',\psi\rangle_{L^2}=0} \langle(\bar u' - \psi) , (- L_{\bar u})^{-1}   (\bar u' - \psi)\rangle_{L^2}\;.
\end{equation}
The above minimum is achieved at 
\[
\bar\psi = \bar u' - \frac{\|\bar u'\|_{L^2}^2}{\langle \bar u', (- L_{\bar u})\bar u'\rangle_{L^2}} L_{\bar u} \bar u'\;,
\]
so that, recalling \eqref{mu1},
\begin{equation}
\label{C*1}
C^* = \frac 12 \langle(\bar u' - \bar\psi) (- L_{\bar u})^{-1}   (\bar u' - \bar\psi)\rangle_{L^2} = \frac{\|\bar u'\|_{L^2}^4}{2\langle\bar u',(-L_{\bar u})\bar u'\rangle_{L^2}} = \frac{1}{4\mu}\;.
\end{equation}

In view of \eqref{Ihk3}, \eqref{Ihk4}, and \eqref{C*}, to conclude the proof of the statement (b), it remains to show that there is $\bar Q$ for which the minimum is obtained, i.e., there exists a solution $\bar Q$ to the linear equation $\bar u''\bar Q'+\frac 12 \bar u' \bar Q'' = \bar \psi$ satisfying \eqref{Q2}. This solution can be explicitly computed, precisely, 
\[
\bar Q(\xi) = 2\int_0^\xi\!\rmd\xi'\, \frac{1}{\bar u'(\xi')^2} \int_{-\infty}^{\xi'}\rmd\xi''\, \bar u'(\xi'')\bar\psi(\xi'')\;,
\]
which satisfies \eqref{Q2} in view of \eqref{asbaru}.

\medskip
\noindent {\it Proof of \rm (a).} Let $\phi_\eps$ be as in \eqref{recovery1}. By \eqref{Sgk1}, \eqref{Jgk1}, and integration by parts we have, for any $H\in C^{1,2}([0,T]\times\bb T^d)$,
\begin{equation}
\label{Jgk2}
\begin{split}
 S_\eps(\phi_\eps)  & \ge \frac 1\eps \int_0^T\!\rmd t \int\!\rmd x\, \left[ \bigg(\partial_t \phi_\eps- \frac 12 \Delta \phi_\eps \bigg)H - \frac 12 \phi_\eps(1-\phi_\eps) |\nabla H|^2 \right]  \\ & \quad +  \frac 1\eps \int_0^T\!\rmd t \int\!\rmd x\, \left( B(\phi_\eps) \frac{1-\rme^H}{\eps^2} + D(\phi_\eps) \frac{1-\rme^{-H}}{\eps^2}\right)\;.
\end{split}
\end{equation}
We choose $H$ of the form $H(t,x) = \eps H_1 (t,x,d(x,\Gamma(t))/ \eps)$, with $H_1\colon [0,T]\times\bb T^d \times \bb R \to \bb R$ a smooth function to be determined later such that, for some $\gamma'>0$,
\begin{equation}
\label{h1}
\sup_{(t,x,\xi) \in [0,T] \times \bb T^d \times \bb R} \rme^{\gamma' |\xi|} \big\{ |H_1(t,x,\xi)| +  |\nabla_xH_1(t,x,\xi)| + |\partial_\xi H_1(t,x,\xi)|\big\}   < \infty\;.
\end{equation}
Noticing that the dependence on $\phi_\eps$  of the integrands in the right-hand side of \eqref{Jgk2}  is locally Lipschitz and recalling the hypothesis on $R_\eps$, in view of the above assumptions on $H_1$, it is readily seen that  the contribution due to $R_\eps$ is $o(1)$ as $\eps\to 0$, and hence it can be neglected. 

Therefore, few direct calculations (using \eqref{expsu}, here applied to $\phi_\eps - \eps R_\eps$, and the co-area formula) give
\[
\begin{split}
\liminf_{\eps\to 0} S_\eps(\phi_\eps) & \ge \int_0^T\!\rmd t \int_{\Gamma(t)}\!\rmd\sigma \int\!\rmd\xi\, \bigg\{ \left[ \bar u' \Big( \partial_t d - \frac 12 \Delta d \Big) + \bar u'' Q' + \frac 12 \bar u' Q'' \right] H_1 \\ & \quad - \frac 12 \bar u(1-\bar u) (H_1')^2 - \frac{B(\bar u)+D(\bar u)}2 H_1^2 \bigg\}\;,
\end{split}
\]
where the notation $Q'=\partial_\xi Q(x,t,\xi)$, $Q''=\partial_{\xi\xi}Q(t,x,\xi)$, $H_1'=\partial_\xi H_1(x,t,\xi)$, and $H_1''=\partial_{\xi\xi}H_1(t,x,\xi)$ has been adopted. The maximum of the expression in the right-hand side is obtained for $H_1 = \mc H$ with, for any $(t,x)\in [0,T]\times\bb T^d$ fixed, $\mc H(t,x,\cdot)$ solution to 
\[
L_{\bar u} \mc H =  \Big(\frac 12 \Delta d -\partial_td\Big) \bar u' - \bar u'' Q' - \frac 12 \bar u' Q'' := F_Q \;,
\]
which satisfies the assumptions \eqref{h1} in view of \eqref{asbaru}, \eqref{Q3}, and \eqref{decayl}.  Hence,
\begin{equation}
\label{fh}
\liminf_{\eps\to 0} S_\eps(\phi_\eps) \ge \frac 12 \int_0^T\!\rmd t \int_{\Gamma(t)}\!\rmd\sigma \int\!\rmd\xi\, F_Q (-L_{\bar u})^{-1}F_Q\;.
\end{equation}
We next observe that, in view of \eqref{C*}, for each $(t,x)\in [0,T]\times\bb T^d$ fixed, 
\[
\frac 12 \int\!\rmd\xi\, F_Q (-L_{\bar u})^{-1}F_Q \ge \Big( \partial_t d - \frac 12 \Delta d \Big)^2 C^*\;.
\]
As $\partial_td(x,\Gamma(t)) - \frac 12 \Delta d(x,\Gamma(t)) = v_t(x) - \frac 12 \kappa_t(x)$ for $x\in \Gamma(t)$, the statement (a) follows by \eqref{fh} and \eqref{C*1}.
\qed

\section{Approximating nucleation events}
\label{sec:5}

In this section we discuss how the nucleation part of the rate function in \eqref{s+s} can be recovered from its absolutely continuous part. The
general result should be the following. Given $T>0$, let $\Gamma=\Gamma(t)$, $t\in [0,T]$, be a path of interfaces satisfying $S(\Gamma)<+\infty$ (with possible nucleation events), then there exists a sequence of paths $\{\Gamma_\delta\}$ with zero nucleation cost such that $\Gamma_\delta\to \Gamma$ and $S_\mathrm{ac}(\Gamma_\delta) \to S(\Gamma)$ as $\delta\to 0$. We shall not discuss the issue at this level of generality but rather provide a strategy for a special class of paths. We also restrict the analysis to the two-dimensional isotropic case.

The basic idea is that in dimension $d=2$ points (i.e., $(d-2)$-dimensional interfaces) can be nucleated with no cost and we can then let them evolve for a short time (vanishing as $\delta\to 0$) in such a way that at the final time the resulting interface approximates the one we want to nucleate. Moreover, as we are going to argue, it is possible to arrange the evolution so that the corresponding cost indeed approximates the nucleation one.

Let us consider a path $\Gamma$ of the form,
\begin{equation}
	\label{ori}
	\Gamma(t)=
	\begin{cases}
		\emptyset  & \textrm{if } t\in [0,\bar t)\;, \\
		\Gamma^0 (t) & \textrm{if } t \in [\bar t, T]\;,
	\end{cases}
\end{equation}
where $\bar t$ is the nucleation time and $\Gamma^0$ is a smooth path of smooth one-dimensional interfaces with initial value $\bar\Gamma := \Gamma^0(\bar t)$. We assume that $\Gamma^0(t)=\partial\Omega^0(t)$, for some open set $\Omega^0(t)$, when $t\in (\bar t,T]$, while $\bar \Gamma = \lim_{t\downarrow \bar t} \Omega^0(t)$. The corresponding nucleation cost is 
\begin{equation*}
	S_\mathrm{nucl} (\Gamma) = 2 \tau \mathrm{Per} (\bar\Gamma)\;,
\end{equation*}
where we recall that $\tau$ is the surface tension and $\mathrm{Per}(\bar\Gamma)$ denotes here the length of $\bar\Gamma$, while the factor $2$ is due to the fact that $\bar\Gamma$ has be thought as an interface with double multiplicity. 

In view of the assumed smoothness of $\bar \Gamma$, by localization, it suffices to consider the case in which it is a segment, say of length $\ell$.  In order to define the corresponding approximating path $\Gamma_\delta$, we first construct a path $\Sigma_\delta(s)$, $s\in [0,\sigma_\delta]$, with $\sigma_\delta\to 0$, satisfying $\Sigma_\delta(0) \to \bar\Gamma$, $\mathrm{Per}\big(\Sigma_\delta(0)\big)\to 2 \, \mathrm{Per} (\bar\Gamma)= 2\ell$, and $\Sigma_\delta(\sigma_\delta)=\emptyset$. To this end, chop the segment $\bar\Gamma$ into $N_\delta$ sub-segments (with $N_\delta$ diverging as $\delta\to 0$) and then fat each subsegment to an ellipse with major axis of length $\ell/N_\delta$ and minor axis of length $m_\delta \ll \ell/N_\delta$. Denoting by $\bar\Sigma_\delta$ the resulting interface, then $\bar\Sigma_\delta\to \bar\Gamma$ and $\mathrm{Per}\big(\bar\Sigma_\delta\big)\to 2 \ell$. The path $\Sigma_\delta(s)$, $s\ge 0$ is now defined as the evolution by mean curvature with initial datum $\bar\Sigma_\delta$ and transport coefficient $\theta$. Here, we understand that each ellipse evolves by mean curvature separately. By comparing the evolution of each ellipse with that of a circle of initial diameter equals to the major axis, we deduce that $\Sigma_\delta(\sigma_\delta)=\emptyset$ for some $\sigma_\delta \le (\ell/N_\delta)^2/(8\theta )$.

We now set
\begin{equation*}
	\Gamma_\delta(t) :=
	\begin{cases}
		\emptyset & t\in [0,\bar t -\sigma_\delta)\;, \\
		\Sigma_\delta(\bar t - t) & t\in [\bar t -\sigma_\delta, \bar t)\;, \\
		\Gamma^0_\delta(t)  & t\in [\bar t, T]\;,
	\end{cases}
\end{equation*}
where $\Gamma^0_\delta$ is a suitable approximation of the path $\Gamma^0$ in \eqref{ori}, satisfying $\Gamma^0_\delta(\bar t) = \Sigma_\delta(0)$, organized so that $S_{\mathrm{ac}, [\bar t, T]}(\Gamma^0_\delta) \to S_{\mathrm{ac}, [\bar t, T]}(\Gamma^0)$,  whose details are omitted. To conclude we next show that 
\begin{equation*}
	S_{\mathrm{ac}, [0,\bar t]} (\Gamma_\delta) \to 2 \tau \ell\,.
\end{equation*}
Even if this is essentially a Friedlin-Wentzel argument for evaluating the quasi-potential in the reversible case, we provide the details of the computation. Denoting by $v_\delta$ and $\kappa_\delta$ the normal velocity and mean curvature of  $\Gamma_\delta$, we write
\begin{equation*}
	\frac 1{4\mu} (v_\delta -\theta \kappa_\delta)^2 
	= \frac 1{4\mu} (v_\delta +\theta \kappa_\delta)^2
	- \frac{\theta}{\mu} \kappa_\delta v_\delta\;.
\end{equation*}
By construction of the path, that has been obtained by time reversal of motion by mean curvature, the first term on the right-hand side above vanishes. Since 
\begin{equation*}
\frac{\rmd}{\rmd t } \mathrm{Per}(\Gamma_\delta(t)) = -\int_{\Gamma_\delta(t)} v_\delta \kappa_\delta\,,
\end{equation*}
we conclude by using the Einstein relation \eqref{theta}.

\appendix
\section{}
\label{app:a}

\begin{lemma}                
\label{lem:Q}                 
For $\hat H$ as in \eqref{cH} there exists a unique solution $\bar Q$ to the equation $L(\bar m' \bar Q) =\hat H$ which satisfies $\bar{Q}(0)=0$ and Eq.\eqref{Q1}.
\end{lemma}

\begin{proof}
Recall the properties of $L$, given by \eqref{L} and acting on $L^2(\bb R,\nu(\rmd\xi))$, described in Subsection \ref{sec:3}. Letting $E:=\{\phi\in L^2(\bb R,\nu(\rmd\xi)) \colon \int\!\nu(\rmd \xi)\; \bar m'(\xi)\phi(\xi) = 0 \}$, the bounded operator $L\colon E\to E$ is symmetric, coercive, and therefore a bijection. As $\hat H\in E$, we deduce that there exists a unique solution $\phi^*\in E$ to the equation $L\phi = \hat H$. This implies that the family of functions of the form $\psi_\lambda := \phi^* + \lambda \bar m'$, $\lambda\in \bb R$, coincides with the set of all the solutions to $L\psi = \hat H$ in $L^2(\bb R,\nu(\rmd\xi))$. Moreover, from the explicit form \eqref{L} of $L$ and the smoothness of $\tilde J$, $\bar m$, $c$, and $\hat H$, the functions $\psi_\lambda$ turn out to be smooth as well. In particular, since $\bar m '>0$, the values of $\lambda$ is uniquely determined by the condition $\psi_\lambda(0) = 0$. 

So far, we have proved that there exists a unique solution $\bar \psi\in L^2(\bb R,\nu(\rmd\xi))$ to the equation $L\psi = \hat H$ which has the form $\bar \psi = \bar m'\bar Q$ with $\bar Q$ a smooth function with $\bar Q(0)=0$. We are left with the proof of Eq.\eqref{Q1}, which is equivalent, view of \eqref{mexp}, to prove that
\begin{equation}
\label{psi1}
\sup_{\xi\in\bb R}\frac{ \rme^{\alpha|\xi|}\big| \bar \psi(\xi) \big| +  \rme^{\alpha|\xi|} \big|\bar \psi'(\xi)\big|}{1+|\xi|}<+\infty\;,
\end{equation}
with $\alpha$ as in \eqref{alpha}.

Recalling \eqref{L}, \eqref{f}, \eqref{H}, and \eqref{cH}, we have that
\begin{equation}
\label{eqpsi}
\bar\psi = p \tilde J * \bar\psi + g\;, \qquad \bar\psi' = p \tilde J * \bar\psi' + g'\,,
\end{equation}
where, see \eqref{fb},
\begin{equation}
\label{barp}
p := \beta (1-m_\beta^2) = \frac{1}{1+f_\beta''(m_\beta)} < 1
\end{equation}  
and
\begin{equation}
\label{g}
g := (m_\beta^2-\bar m^2)  \beta \tilde J *\bar\psi - (1-\bar m^2) \beta f + \frac{\theta\sqrt{1-\bar m^2}}{2\bar a}\bar m'\;.
\end{equation}
Since $\bar\psi\in L^2(\bb R,\nu(\rmd\xi))$, both $\tilde J*\bar\psi$ and $\tilde J' *\bar\psi$ are bounded functions, so that, in view of \eqref{mexp} and \eqref{f}, 
\begin{equation}
\label{gst}
\sup_{\xi\in\bb R} \rme^{\alpha|\xi|}\big(\big| g(\xi) \big| +\big|g'(\xi)\big|\big) < +\infty\;.
\end{equation}
This immediately implies, by \eqref{eqpsi}, that also $\bar\psi$ and $\bar\psi'$ are bounded functions. To obtain the decay properties \eqref{psi1}, we now adapt to the present context part of the analysis developed in \cite{DGP} to study the spatial structure of the traveling fronts. Actually, we only show that 
\begin{equation}
\label{psi2}
\sup_{\xi >0}\frac{\rme^{\alpha\xi}\big| \bar \psi(\xi)\big|}{1+\xi}<+\infty\;,
\end{equation}
since the proof of the other bound in \eqref{psi1} is similar. 

Let
\begin{equation}
\label{K}
K(\xi,\xi') := \rme^{\alpha(\xi-\xi')}p\tilde J(\xi-\xi')\;, \qquad K_\mathrm{o}(\xi,\xi') := K(\xi,\xi') \id_{\xi'>0}\;.
\end{equation}
By \eqref{alpha}, $K(\xi,\xi')$ is a probability kernel, i.e., $\int\!\rmd\xi'\;K(\xi,\xi') = 1$. By \eqref{eqpsi}, for any $\xi>0$,
\begin{equation*}
\bar\psi(\xi) = \int_0^{+\infty}\!\rmd\xi'\;\rme^{-\alpha(\xi-\xi')}K_\mathrm{o}(\xi,\xi') \bar\psi(\xi') + \int_{-1}^0\!\rmd\xi'\;\rme^{-\alpha(\xi-\xi')}K(\xi,\xi') \bar\psi(\xi') +  g(\xi)\;,
\end{equation*}
which implies, by iteration, 
\begin{equation}
\label{psin}
\begin{split}
	\bar\psi(\xi) & = \sum_{j=0}^{n-1} \int_0^{+\infty}\!\rmd\xi'\; \rme^{-\alpha(\xi-\xi')} K_\mathrm{o}^j(\xi,\xi') \int_{-1}^0\!\rmd\xi''\; \rme^{-\alpha(\xi'-\xi'')}K(\xi',\xi'') \bar\psi(\xi'') \\ & \quad + \int_0^{+\infty}\!\rmd\xi'\; \rme^{-\alpha(\xi-\xi')} K_\mathrm{o}^n(\xi,\xi') \bar\psi(\xi')  \\ & \quad + \sum_{j=0}^{n-1} \int_0^{+\infty}\!\rmd\xi'\; \rme^{-\alpha(\xi-\xi')} K_\mathrm{o}^j(\xi,\xi') g(\xi')  \;, \qquad \forall\, \xi\ge 0 \quad \forall\, n\in\bb N\;.
\end{split}
\end{equation}
Above, the iterated kernel $K_\mathrm{o}^j(\xi,\xi')$ is recursively defined by $K_\mathrm{o}^0(\xi,\xi') = \delta(\xi-\xi')$, $K_\mathrm{o}^1(\xi,\xi')  =  K_\mathrm{o}(\xi,\xi')$, and $K_\mathrm{o}^j(\xi,\xi') = \int\!\rmd\xi'' K_\mathrm{o}(\xi,\xi'') K_\mathrm{o}^{j-1}(\xi'',\xi')$ for $j>1$.
    
Since $\bar\psi$ and $g$ are bounded functions and $\int\!\rmd\xi\;\tilde J(\xi) = 1$, the $j$-th terms of the sums in the right-hand side of \eqref{psin} are bounded by a constant multiple of $p^j$, and the term in the middle line by a constant multiple of $p^n$. As $p<1$ we thus have, letting $n\to \infty$ in \eqref{psin}, 
\begin{equation}
\label{green1}
\rme^{\alpha\xi}\bar\psi(\xi) = \int_{-1}^0\!\rmd\xi'\;  \pi(\xi,\xi') \rme^{\alpha\xi'} \bar\psi(\xi') + G(\xi)\;, \quad \xi>0\;,
\end{equation}
where both the series
\begin{equation*}
\pi(\xi,\xi') := \sum_{j=0}^\infty \int_0^{+\infty}\!\rmd\xi''\; K_\mathrm{o}^j(\xi,\xi'') K(\xi'',\xi') 
\end{equation*}
and 
\begin{equation*}
G(\xi) :=  \sum_{j=0}^\infty \int_0^{+\infty}\!\rmd\xi'\; K_\mathrm{o}^j(\xi,\xi')  \rme^{\alpha\xi'}g(\xi')
\end{equation*}
converge.

The Green function $\pi(\xi,\xi')$, $\xi> 0$, $\xi'\in [-1,0]$, can be interpreted as a probability kernel because it is non negative and
\begin{equation}
\label{pi1}
\int_{-1}^0\!\rmd\xi'\; \pi(\xi,\xi') = 1 \quad \forall\, \xi>0\;.
\end{equation}
Moreover, there exists a probability density $\varrho(\xi)$, $\xi\in [-1,0]$, so that, for any function $\varphi\in C( [-1,0])$,
\begin{equation}
\label{pi2}
\lim_{\xi\to+\infty}	\int_{-1}^0\!\rmd\xi'\; \pi(\xi,\xi') \varphi(\xi') =	\int_{-1}^0\!\rmd\xi'\; \varrho(\xi') \varphi(\xi')\;.
\end{equation}
We omit the proof of \eqref{pi1} and \eqref{pi2} which are a special case of \cite[Eq.(3.69) and Lemma 3.7]{DGP}. 

Since $\bar\psi$ is a continuous function, from \eqref{pi2} the integral in the right-hand side of  \eqref{green1} converges to $\int_{-1}^0\!\rmd\xi'\; \varrho(\xi') \rme^{\alpha\xi'} \bar\psi(\xi')$ as $\xi\to +\infty$. Therefore, to prove \eqref{psi2} it remains to show that 
\begin{equation}
\label{G1}
\sup_{\xi> 0}\frac{\big|G(\xi)\big|}{1+\xi}<+\infty\;.
\end{equation}
By \eqref{gst} there is $C>0$ such that
\begin{equation*}
|G(\xi)| \le C + C \sum_{n=1}^\infty \int_0^{+\infty}\!\rmd\xi'\; K_\mathrm{o}^n(\xi,\xi')\;.
\end{equation*}
To estimate the $n$-th term of the sum in the right-hand side the key observation is taken from the proof of \cite[Eq.(3.69)]{DGP}. We have,
\begin{equation*}
\begin{split}
	\int_0^{+\infty}\!\rmd\xi'\; K_\mathrm{o}^n(\xi,\xi') & \le \int\!\rmd\xi_1 \cdots\rmd\xi_n\; K(\xi,\xi_1)\cdots K(\xi_{n-1},\xi_n)\, \id_{\xi_n>0} =: I_n\;.
\end{split}
\end{equation*}
Since the probability kernel $K(\xi,\xi')$ depends only on the difference $\xi'-\xi$, see \eqref{K}, the multiple integral $I_n$ can be viewed as an expectation with respect to $n$ i.i.d.\ random variables $Y_j=\xi_j-\xi_{j-1}$, $j=1,\ldots, n$, where $\xi_0:=\xi$, each one with the distribution of $\xi'-\xi$ as given by $K(\xi,\xi')\,\rmd\xi'$. More precisely, as $\{\xi_n >0\} = \{ (\xi_n-\xi_{n-1}) +  (\xi_{n-1}-\xi_{n-2}) + \cdots + (\xi_1-\xi) > -\xi \}$ and observing that, in view of \eqref{K}, $\bb E(Y_1) = \int\!\rmd\xi'\; K(\xi,\xi')(\xi'-\xi) = C_1 < 0$,
\begin{equation*}
\begin{split}
	I_n & = \bb P\bigg(\sum_{j=1}^nY_j> -\xi\bigg)	= \bb P\bigg( \sum_{j=1}^n (Y_j-C_1) > n |C_1| -\xi \bigg)\;.
\end{split}
\end{equation*}
If $n\le 2\xi/|C_1|$ we use the obvious estimate $I_n\le 1$, while for $n> 2\xi/|C_1|$,  by Chebyshev's inequality,
\begin{equation*}
I_n \le 	\bb P\bigg( \sum_{j=1}^n (Y_j-C_1) > \frac n2 |C_1| \bigg) \le \frac{16}{C_1^4n^4}\bb E\bigg[\big( \sum_{j=1}^n (Y_j-C_1) \big)^4\bigg] \le\frac{16C_4 n + 96 C_2 n^2}{C_1^4n^4}\;,
\end{equation*} 
where $C_2 := \bb E\big[(Y_1 -C_1)^2\big] $ and  $C_4 := \bb E\big[(Y_1 -C_1)^4\big] $. We conclude that 
\begin{equation*}
|G(\xi)| \le C + C \sum_{n\ge 1}^\infty I_n \le C +\frac{2 C}{|C_1|} \xi + \sum_{n>0}\frac{16C_4 n + 96 C_2 n^2}{C_1^4n^4}\;,
\end{equation*}
from which \eqref{G1} follows.
\end{proof}

\medskip
\begin{proof}[Proof of Eq.\eqref{Htilde}]
We write $H_\eps=\mc H_\eps+\mc K_\eps$ with $\mc H_\eps$ solving the linear equation,
\begin{equation}
	\label{mcH}
	- \nabla \cdot (\phi_\eps (1-\phi_\eps) \nabla \mc H_\eps) +\frac{B(\phi_\eps)+D(\phi_\eps)}{\eps^2} \mc H_\eps =  \partial_t \phi_\eps -\frac12 \Delta \phi_\eps - \frac{B(\phi_\eps)-D(\phi_\eps)}{\eps^2}\;,
\end{equation}
and therefore $\mc K_\eps$ satisfying
\begin{equation}
	\label{mcK}
	- \nabla \cdot (\phi_\eps (1-\phi_\eps) \nabla \mc K_\eps) +\frac{B(\phi_\eps)  (\rme^{\mc H_\eps+\mc K_\eps} -1 - \mc H_\eps)  -D(\phi_\eps) (\rme^{-\mc H_\eps-\mc K_\eps} -1 + \mc H_\eps)}{\eps^2} =  0  \; .
\end{equation}
Since $B(\phi_\eps)+D(\phi_\eps)$ and $\phi_\eps(1-\phi_\eps)$ are strictly positive, the first equation has a unique solution in $L^2(\bb T^d)$, which is a smooth function of $(t,x)$ by the smoothness of $\phi_\eps$ and elliptic regularity. In addition, in view of the expansions \eqref{expsu} and \eqref{expsu1} (the latter for $H=0$), the right-hand side in \eqref{mcH} is $O(\eps^{-1})$, so that the maximum principle yields $\mc H_\eps = O(\eps)$.

Combining \eqref{nuova} with \eqref{mcH}, by \eqref{expsu} and \eqref{expsu1} (the latter for $H=0$),
\begin{equation}
	\label{hk1}
	- \nabla \cdot (\phi_\eps (1-\phi_\eps) \nabla (\mc H_\eps-\eps H_1)) +\frac{B(\phi_\eps)+D(\phi_\eps)}{\eps^2} (\mc H_\eps-\eps H_1) = O(1)\;,
\end{equation}
hence the maximum principle yields $\mc H_\eps-\eps H_1 = O(\eps^2)$.

Next, using \eqref{mcK} we show that $\mc K_\eps = O(\eps^2)$, whence $\tilde H_\eps = \eps^{-2} \mc K_\eps + \eps^{-2}(\mc H_\eps-\eps H_1) = O(1)$. Indeed, recalling that $B(\phi_\eps)$ and  $D(\phi_\eps)$ are strictly positive and noticing that the non-linear term in \eqref{mcK} is an increasing function of $\mc K_\eps$, by comparison principle it is enough to construct super- and sub-solutions in the form $\mc K_\eps^\pm = \pm C\eps^2$, where $C$ is a (suitably large) positive constant. This simply follows from the elementary inequality $1 - \mc H - \rme^{C\eps^2-\mc H} >0$, which holds if $|\mc H| \le C_1\eps$ for given $C_1>0$, provided $C$ is large enough and $\eps$ is small enough. 

We are left with the estimate on the gradient of $\tilde H_\eps$. To this end, we first notice that, as $\mc H_\eps = O(\eps)$ and $\mc K_\eps = O(\eps^2)$, Eq.\eqref{mcK} can be recast in the form,
\begin{equation}
	\label{hk2}
	- \nabla \cdot (\phi_\eps (1-\phi_\eps) \nabla \mc K_\eps +\frac{B(\phi_\eps)+D(\phi_\eps)}{\eps^2} \mc K_\eps = O(1)\;.
\end{equation}
By \eqref{hk1} and \eqref{hk2} we get,
\begin{equation}
	\label{hk3}
	- \nabla \cdot (\phi_\eps (1-\phi_\eps) \nabla \tilde H_\eps +\frac{B(\phi_\eps)+D(\phi_\eps)}{\eps^2} \tilde H_\eps = O(\eps^{-2})\;.
\end{equation}
As $H_\eps = O(1)$, the estimate on $\eps|\nabla \tilde H_\eps|$ follows by a standard covering argument with balls of radius $\eps$ and applying elliptic regularity.	
\end{proof}

\medskip

\subsection*{Acknowledgements}
Section \ref{sec:4} is our attempt to answer questions raised by Claudio Landim.

\end{document}